\documentclass[10pt, journal]{IEEEtran}
\usepackage{amsmath,amssymb,amsfonts,amsthm}
\usepackage{graphicx}
\usepackage{epstopdf}
\usepackage[short]{optidef}
\usepackage{algorithmicx,algorithm}
\usepackage[noend]{algpseudocode}
\usepackage{subfig}
\usepackage{arydshln}
\usepackage{mathtools}
\usepackage{color}
\DeclareMathOperator{\diag}{diag}
\DeclareMathOperator{\Tr}{tr}

\DeclareMathOperator{\Vector}{vec}
\DeclareMathOperator{\pdf}{pdf}
\DeclareMathOperator{\Res}{Res}
\DeclareMathOperator{\Expectation}{\mathop{\mathbb{E}}}

\newtheorem{Proposition}{Proposition}

\newtheorem{Lemma}{Lemma}
\allowdisplaybreaks[4]
\begin{document}
\title{Large Intelligent Surface Aided \\Physical Layer Security Transmission}
\author{\IEEEauthorblockN{Biqian Feng, Yongpeng Wu, Mengfan Zheng, Xiang-Gen Xia, \\Yongjian Wang, and Chengshan Xiao}
	\thanks{B. Feng, and Y. Wu are with the Department of Electronic Engineering, Shanghai Jiao Tong University, Minhang 200240, China (e-mail: fengbiqian@sjtu.edu.cn; yongpeng.wu@sjtu.edu.cn) (Corresponding author: Yongpeng Wu).}
	\thanks{M. Zheng is with the Department of Electrical and Electronic Engineering, Imperial College London, London SW7 2AZ, U.K. (email: m.zheng@imperial.ac.uk).
	}
	\thanks{X.-G. Xia is with the Department of Electrical and Computer Engineering, University of Delaware, Newark,
		DE 19716, USA. (e-mail: xxia@ee.udel.edu).}
	\thanks{Y. Wang is with The National Computer Network Emergency Response Technical Team. Coordination Center of China, Chaoyang 100029, China (e-mail: wyj@cert.org.cn) (Corresponding author: Yongjian Wang).}
	\thanks{C. Xiao is with the Department of Electrical and Computer Engineering, Lehigh University, Bethlehem, PA 18015 USA (e-mail: xiaoc@lehigh.edu).}
}
\maketitle
\begin{abstract}
	In this paper, we investigate a large intelligent surface-enhanced (LIS-enhanced) system, where a LIS is deployed to assist secure transmission. Our design aims to maximize the achievable secrecy rates in different channel models, i.e., Rician fading and (or) independent and identically distributed Gaussian fading for the legitimate and eavesdropper channels. In addition, we take into consideration an artificial noise-aided transmission structure for further improving system performance.
	The difficulties of tackling the aforementioned problems are the structure of the expected secrecy rate expressions and the non-convex phase shift constraint. To facilitate the design, we propose two frameworks, namely the sample average approximation based (SAA-based) algorithm and the hybrid stochastic projected gradient-convergent policy (hybrid SPG-CP) algorithm, to calculate the expectation terms in the secrecy rate expressions. Meanwhile, majorization minimization (MM) is adopted to address the non-convexity of the phase shift constraint.
	In addition, we give some analyses on two special scenarios by making full use of the expectation terms. Simulation results show that the proposed algorithms effectively optimize the secrecy communication rate for the considered setup, and the LIS-enhanced system greatly improves secrecy performance compared to conventional architectures without LIS.
\end{abstract}
\section{Introduction}
With the popularizing of user devices, a variety of wireless technologies have been proposed to improve both spectrum efficiency (SE) and energy efficiency (EE) of wireless networks. Recently, a novel concept of large intelligent surface (LIS) has been introduced as a promising technique due to its capability of achieving high SE and EE. It is convenient to control the beamforming design at the access point (AP) and phase shifts at the LIS dynamically according to the changes in the environment. Reference \cite{Survey} summarizes four specific benefits of LIS-enhanced wireless communication systems as follows: (i) easy deployment and sustainable operations; (ii) flexible reconfiguration via passive beamforming; (iii) enhanced capacity and SE/EE performance; (iv) exploration of emerging wireless applications. Many works have emerged to generalize classical scenarios to LIS-enhanced systems, such as channel estimation \cite{Channel Estimation1}-\cite{Channel Estimation4} and unmanned aerial vehicles \cite{UAV1}-\cite{UAV2}, and to verify the effectiveness of LIS-enhanced systems compared with conventional ones \cite{Q.Wu1}-\cite{C. Yuen3}.

\textbf{Related Works:} Recently, LIS-enhanced systems have been introduced into the physical layer security design in wireless communications. Reference \cite{Secure_IRS1} first proposes the security issues in LIS-enhanced systems and solves the beamforming and phase shift problem efficiently with both block coordinate decent and majorization minimization (MM) techniques \cite{MM_Proposed} under the conditions of multi-input, single-output, single-eavesdropper (MISOSE) and perfect channel state information (CSI) of both legitimate and eavesdropper channels. Reference \cite{Secure_IRS3} develops suboptimal solutions for both beamforming and phase shifts with semidefinite relaxation (SDR) and Gaussian randomization methods with direct links existing in the environment. Reference \cite{Secure_IRS2} simplifies the phase shift scheme with MM in the network model of one legitimate receiver and one eavesdropper and applies it to the multiple-antenna eavesdropper case. The model of multiple legitimate receivers is investigated in \cite{Secure_IRS4}, where manifold optimization is applied to handle the phase shifts and successive convex approximation (SCA) method is used for beamforming and artificial noise (AN) injection. Meanwhile, the model of multiple eavesdroppers is investigated in \cite{Secure_IRS5}, where SDR method is applied for beamforming, AN and phase shifts. Reference \cite{Secure_IRS6} studies the maximization of the minimum secrecy rate among several legitimate users in the presence of multiple eavesdroppers, where the unit modulus constraints of the phase shifts are approximated by a set of convex constraints. References \cite{Secure_IRSs1}-\cite{Secure_IRSs2} investigate the system with multi-input, multi-output, multiple-antenna eavesdropper (MIMOME). In summary, references \cite{Secure_IRS1}-\cite{Secure_IRSs2} focus on a series of schemes under the assumption of perfect CSI of all channels. Reference \cite{Secure_IRS7} takes into account the model of multiple single-antenna legitimate receivers, which do not have line-of-sight (LoS) communication links, in the presence of multiple multi-antenna potential eavesdroppers whose CSI is not perfectly known. It is worth noting that reference \cite{Secure_IRS7} solves the phase shift constraint with norm-difference method instead of the aforementioned MM, SCA or SDR. Reference \cite{B.Feng} tries to take into consideration the model of rank-one AP-LIS channel and statistical LIS-Receiver/Eavesdropper channels and eliminates the effect of phase shifts for statistical CSI channel coefficients.

\textbf{Main Contributions:} In contrast to the previous literature, this paper studies a LIS-enhanced MISO system with Rician channel in the AP-LIS link and independent and identically distributed (i.i.d.) Gaussian fading channel in the LIS-eavesdropper link due to the challenge of acquiring perfect CSI of the eavesdropping channels at the AP. The main contributions of this paper can be summarized as follows.
\begin{itemize}
	\item To the best of our knowledge, this is the first work to explore the use of the LIS to enhance the physical layer secret communication rate under the condition of i.i.d. Gaussian fading channel in the LIS-eavesdropper link. We formulate four achievable secrecy rates under different assumptions and jointly optimize the AN-aided beamforming at the AP and phase shifts at the LIS.
	\item The problem is quite challenging due to the following reasons. First, the non-convex phase shift constraint makes the problem essentially an NP-hard problem. Second, computing the expected secrecy rate expressions is computationally expensive. In view of these problems, we adopt the MM algorithm to tackle the phase shift constraint. Meanwhile, we propose two algorithms, a sample average approximation based (SAA-based) algorithm and a hybrid stochastic projected gradient-convergent policy (hybrid SPG-CP) algorithm, to convert the expectation operations to the determined structure at each iteration. Besides, we generalize the SPG-based algorithm and prove that the expectation of the projected gradient approaches \begin{small}$0$\end{small} as the number of iterations approaches infinity.
	\item For the case of Rician fading channel at legitimate receiver and i.i.d. Gaussian fading channel at eavesdropper, we develop an alternating optimization method to solve the problem more efficiently based on the exact calculation of the expectation function. For the case of i.i.d. Gaussian fading channels at both legitimate receiver and single-antenna eavesdropper sides, we verify that it is unnecessary to take the AN-aided structure to suppress the single-antenna eavesdropper.
	\item Finally, simulation results validate the effectiveness of the proposed algorithms. It can be seen that the LIS dramatically improves the quality of the whole secrecy system. Furthermore, the larger number of LIS elements we use, the better performance the system achieves.
\end{itemize}

The rest of this paper is organized as follows. Section \uppercase\expandafter{\romannumeral2} introduces the secure transmission model, achievable secrecy rate and problem formulation. Section \uppercase\expandafter{\romannumeral3} develops two efficient algorithms for the stochastic optimization problem of secrecy rates. Section \uppercase\expandafter{\romannumeral4} provides analysis on some specific cases. Section \uppercase\expandafter{\romannumeral5} shows some simulation results to evaluate the performances of the proposed algorithms. Section \uppercase\expandafter{\romannumeral6} concludes the paper.

The notations used in this paper are as follows. Boldface lowercase and uppercase letters, such as \begin{small}$\mathbf{a}$\end{small} and \begin{small}$\mathbf{A}$\end{small}, are used to represent vectors and matrices, respectively. \begin{small}$\mathbf{I}_n$\end{small} denotes the \begin{small}$n$\end{small}-by-\begin{small}$n$\end{small} identity matrix. Superscripts \begin{small}$T$\end{small}, \begin{small}$*$\end{small}, and \begin{small}$H$\end{small} stand for the transpose, conjugate, and conjugate transpose, respectively. \begin{small}$\nabla_{\mathbf{X}} f$\end{small} denotes the gradient of \begin{small}$f$\end{small} with respect to \begin{small}$\mathbf{X}$\end{small} and \begin{small}$\nabla_{\mathbf{X}^*} f = \nabla_{\mathbf{X}^*} f^*=(\nabla_{\mathbf{X}} f)^*$\end{small} with a real-valued function \begin{small}$f$\end{small}. \begin{small}$\Vert \mathbf{a}\Vert_*$\end{small} denotes the \begin{small}$\ell_*$\end{small}-norm of the complex vector \begin{small}$\mathbf{a}$\end{small}. \begin{small}$\Vert \mathbf{a}\Vert$\end{small} and \begin{small}$\Vert \mathbf{A}\Vert$\end{small}, respectively, denote the \begin{small}$\ell_2$\end{small}-norm of the complex vector \begin{small}$\mathbf{a}$\end{small} and the Frobenius norm of the complex matrix \begin{small}$\mathbf{A}$\end{small}. \begin{small}$\lambda_{max}\left(\mathbf{A}\right)$\end{small} and \begin{small}$\gamma_{max}\left(\mathbf{A}\right)$\end{small}, respectively, denote the maximum eigenvalue of matrix \begin{small}$\mathbf{A}$\end{small} and its corresponding eigenvector. \begin{small}$\arg\left(\mathbf{v}\right)$\end{small} denotes the phases of complex elements in the vector \begin{small}$\mathbf{v}$\end{small}. \begin{small}$\mathcal{CN}(\mathbf{\mu},\mathbf{\Sigma})$\end{small} denotes a complex circular Gaussian distribution with mean \begin{small}$\mathbf{\mu}$\end{small} and covariance \begin{small}$\mathbf{\Sigma}$\end{small}. \begin{small}$\mathbb{R}(a)$\end{small} denotes the real part of a complex value \begin{small}$a$\end{small}. The inner product \begin{small}$\langle\bullet,\,\bullet\rangle : \mathbb{C}^{n\times n}\times\mathbb{C}^{n\times n}\rightarrow\mathbb{R}$\end{small} is defined as \begin{small}$\langle\mathbf{A},\,\mathbf{B}\rangle = \mathbb{R}\left\{\Tr\left(\mathbf{A}^H \mathbf{B}\right)\right\}$\end{small}. \begin{small}$\diag\left(\mathbf{A}\right)$\end{small} denotes a vector whose elements are extracted from the main diagonals of matrix \begin{small}$\mathbf{A}$\end{small}. \begin{small}$a\perp b$\end{small} denotes \begin{small}$a\cdot b=0$\end{small}. 
\begin{small}$\lceil a \rceil$\end{small} represents the smallest integer greater than or equal to \begin{small}$a$\end{small}.

\section{Secure Transmission Model and Problem Formulation}
\subsection{Secure Transmission Model}
As shown in Fig. \ref{Gaussian MISO wiretap channel with IRS}, we consider a MISO wiretap channel model with a LIS-enhanced link. In this system, there are an AP with \begin{small}$N_t$\end{small} antennas, a LIS with \begin{small}${N_I}$\end{small} passive and low-cost reflecting elements, a single-antenna legitimate receiver and an eavesdropper equipped with \begin{small}$N_e$\end{small} antennas. We assume that the direct signal paths between the AP and the legitimate receiver/eavesdropper are neglected due to unfavorable propagation conditions in our model.
\begin{figure}[!htb]
	\centering
	\includegraphics[scale=0.35]{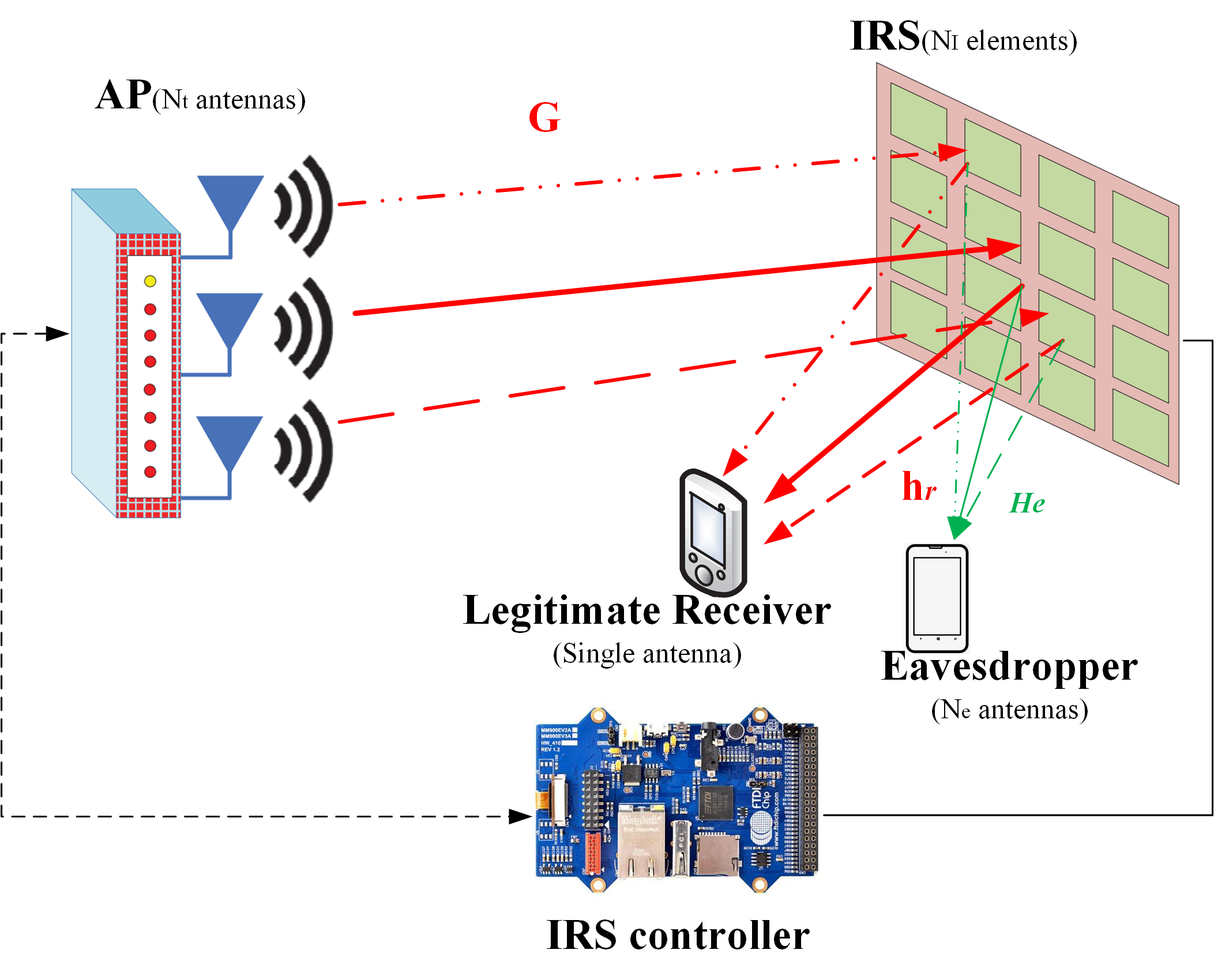}
	\caption{Gaussian MISO wiretap channel with LIS}
	\label{Gaussian MISO wiretap channel with IRS}
\end{figure}

Compared with the traditional wiretap channel model, the LIS-enhanced system introduces a LIS device, which is an intelligent control system that can dynamically adjust the phase through passive beamforming according to the changes in the environment to upgrade the communication quality. The baseband equivalent channels from AP to LIS, from LIS to the legitimate receiver and from LIS to the eavesdropper are, respectively, denoted by \begin{small}$\mathbf{G} \in \mathbb{C}^{N_t \times {N_I}}$\end{small}, \begin{small}$\mathbf{h}_r \in \mathbb{C}^{{N_I} \times 1}$\end{small} and \begin{small}$\mathbf{H}_e \in \mathbb{C}^{{N_I} \times N_e}$\end{small}.
Then, the received complex baseband signals at the legitimate receiver and the eavesdropper are, respectively, given by
\begin{equation}\small
\begin{aligned}
&y_r = \sqrt{\rho_r}\mathbf{h}_r^H \mathbf{\Theta} \mathbf{G}^H\mathbf{x} + n_r\\
&\mathbf{y}_e = \sqrt{\rho_e}\mathbf{H}_e^H \mathbf{\Theta} \mathbf{G}^H\mathbf{x} + \mathbf{n}_e
\end{aligned}
\end{equation}
where \begin{small}$\rho_r$\end{small} and \begin{small}$\rho_e$\end{small} denote the signal-noise-ratios (SNR) at the transmitter for the legitimate receiver and eavesdropper, respectively; \begin{small}$\mathbf{\Theta} \triangleq \diag\left\{e^{j\theta_1}, e^{j\theta_2},...,e^{j\theta_{N_I}}\right\}$\end{small} denotes the phase shift matrix with \begin{small}$\theta_n$\end{small} being the phase shift introduced by the \begin{small}$n$\end{small}th element of LIS, and \begin{small}$n_r$\end{small}, \begin{small}$\mathbf{n}_e$\end{small} are additive white Gaussian noise with variance one. We apply the linear channel prefixing and Gaussian signaling as
\begin{equation}\small
\begin{aligned}
\mathbf{x} = \mathbf{s} + \mathbf{z}
\end{aligned}
\end{equation}
where \begin{small}$\mathbf{s} \sim \mathcal{CN}(\mathbf{0}, \mathbf{\Sigma}_s)$\end{small} and \begin{small}$\mathbf{z} \sim \mathcal{CN}(\mathbf{0}, \mathbf{\Sigma}_z)$\end{small} are independent vectors to convey the message and AN, respectively. Hence, the case of \begin{small}$\mathbf{x} = \mathbf{s}$\end{small} and \begin{small}$\mathbf{z}=\mathbf{0}$\end{small} implies the non-AN aided system.

\subsection{Transmission Schemes and Secrecy Rates}
In this paper, we will discuss transmission schemes from the following three aspects:
\begin{itemize}
\item \textbf{Only i.i.d. Gaussian fading channel of the multiple-antenna eavesdropper.}
\item \textbf{Rician channel/i.i.d. Gaussian fading channel of the legitimate receiver.}
\item \textbf{AN\footnote{We consider the generalized AN scheme, in which one may inject AN to the direction of the message \cite{GAN}.}-aided/non-AN transmission structures.}
\end{itemize}

According to the above three assumptions, we derive the achievable secrecy rates in the following four cases \cite{statistical_channel}\cite{rate1}:
\begin{itemize}
	\item non-AN, the Rician fading channel at receiver and i.i.d. Gaussian fading channel at eavesdropper known to AP and LIS:
	\begin{equation}\small
	\label{rate1}
	\begin{aligned}
	&C_1(\mathbf{\Sigma}_s, \mathbf{\Theta})=\log(1+\rho_r\mathbf{h}_r^H \mathbf{\Theta} \mathbf{G}^H\mathbf{\Sigma}_s \mathbf{G} \mathbf{\Theta}^H \mathbf{h}_r)\\
	&\quad-\Expectation_{\mathbf{H}_e} \left \{\log[\det ( \mathbf{I}+\rho_e\mathbf{H}_e^H \mathbf{\Theta} \mathbf{G}^H\mathbf{\Sigma}_s \mathbf{G} \mathbf{\Theta}^H \mathbf{H}_e)]\right\}.
	\end{aligned}
	\end{equation}

	\item non-AN, both i.i.d. Gaussian fading channels at receiver and eavesdropper known to AP and LIS:
	\begin{equation}\small
	\label{rate2}
	\begin{aligned}
	&C_2(\mathbf{\Sigma}_s, \mathbf{\Theta}) = \Expectation_{\mathbf{h}_r} \left[\log\left(1+\rho_r\mathbf{h}_r^H \mathbf{\Theta} \mathbf{G}^H\mathbf{\Sigma}_s \mathbf{G} \mathbf{\Theta}^H \mathbf{h}_r\right)\right]\\
	&\quad-\Expectation_{\mathbf{H}_e} \left \{\log[\det ( \mathbf{I}+\rho_e\mathbf{H}_e^H \mathbf{\Theta} \mathbf{G}^H\mathbf{\Sigma}_s \mathbf{G} \mathbf{\Theta}^H \mathbf{H}_e)]\right\}.
	\end{aligned}
	\end{equation}
	
	\item AN-aided, the Rician fading channel at receiver and i.i.d. Gaussian fading channel at eavesdropper known to AP and LIS:
	\begin{equation}\small
	\label{rate3}
	\begin{aligned}
	&C_3(\mathbf{\Sigma}_s, \mathbf{\Sigma}_z, \mathbf{\Theta})=\\
	&\log\left(1+\frac{\rho_r\mathbf{h}_r^H \mathbf{\Theta} \mathbf{G}^H\mathbf{\Sigma}_s \mathbf{G} \mathbf{\Theta}^H \mathbf{h}_r}{1+\rho_r\mathbf{h}_r^H \mathbf{\Theta} \mathbf{G}^H\mathbf{\Sigma}_z \mathbf{G} \mathbf{\Theta}^H \mathbf{h}_r}\right)-\\
	&\Expectation_{\mathbf{H}_e} \left\{\log\left[\det \left( \mathbf{I}+\rho_e\mathbf{H}_e^H \mathbf{\Theta} \mathbf{G}^H\left(\mathbf{\Sigma}_s+\mathbf{\Sigma}_z\right) \mathbf{G} \mathbf{\Theta}^H \mathbf{H}_e\right)\right]\right\}\\
	&+\Expectation_{\mathbf{H}_e} \left\{\log\left[\det \left( \mathbf{I}+\rho_e\mathbf{H}_e^H \mathbf{\Theta} \mathbf{G}^H\mathbf{\Sigma}_z \mathbf{G} \mathbf{\Theta}^H \mathbf{H}_e\right)\right]\right\}.
	\end{aligned}
	\end{equation}
	
	\item AN-aided, both i.i.d. Gaussian fading channels at receiver and eavesdropper known to AP and LIS:
	\begin{equation}\small
	\label{rate4}
	\begin{aligned}
	&C_4(\mathbf{\Sigma}_s, \mathbf{\Sigma}_z, \mathbf{\Theta})=\\
	&\Expectation_{\mathbf{h}_r} \left[\log\left(1+\frac{\rho_r\mathbf{h}_r^H \mathbf{\Theta} \mathbf{G}^H\mathbf{\Sigma}_s \mathbf{G} \mathbf{\Theta}^H \mathbf{h}_r}{1+\rho_r\mathbf{h}_r^H \mathbf{\Theta} \mathbf{G}^H\mathbf{\Sigma}_z \mathbf{G} \mathbf{\Theta}^H \mathbf{h}_r}\right)\right]-\\
	&\Expectation_{\mathbf{H}_e} \left\{\log\left[\det \left( \mathbf{I}+\rho_e\mathbf{H}_e^H \mathbf{\Theta} \mathbf{G}^H\left(\mathbf{\Sigma}_s+\mathbf{\Sigma}_z\right) \mathbf{G} \mathbf{\Theta}^H \mathbf{H}_e\right)\right]\right\}\\ &+\Expectation_{\mathbf{H}_e} \left\{\log\left[\det \left( \mathbf{I}+\rho_e\mathbf{H}_e^H \mathbf{\Theta} \mathbf{G}^H\mathbf{\Sigma}_z \mathbf{G} \mathbf{\Theta}^H \mathbf{H}_e\right)\right]\right\}.
	\end{aligned}
	\end{equation}
\end{itemize}

\subsection{Problem Formulation}
In this paper, our goal is to design an effective scheme to maximize the achievable secrecy rates by adjusting the beamforming, AN spatial covariance and phase shifts. The corresponding optimization problems in non-AN and AN-cases are, respectively, formulated as:
	\begin{subequations}\small
		\begin{align}
		(P1)\,\,&\mathop{\max}_{\mathbf{\Sigma}_s,\mathbf{\Theta}}&& C_i(\mathbf{\Sigma}_s, \mathbf{\Theta}) (i=1,2)\\
		&\mathop{\text{s.t.}}&& \Tr\left(\mathbf{\mathbf{\Sigma}_s}\right)\leq 1\\
		&&&\mathbf{\Sigma}_s \succeq 0\\
		&&&\theta_n\in [-\pi, \pi),n=1,\ldots,{N_I}.
		\end{align}
	\end{subequations}
	\begin{subequations}\small
		\begin{align}
		(P2)\,\,&\mathop{\max}_{\mathbf{\Sigma}_{s},\mathbf{\Sigma}_{z},\mathbf{\Theta}}&& C_i(\mathbf{\Sigma}_s, \mathbf{\Sigma}_z, \mathbf{\Theta}) (i=3,4)\\
		&\mathop{\text{s.t.}}&& \Tr\left(\mathbf{\mathbf{\Sigma}_s + \mathbf{\Sigma}_z }\right)\leq 1\label{constraint1}\\
		&&&\mathbf{\Sigma}_s \succeq 0,\,\mathbf{\Sigma}_z \succeq 0\label{constraint2}\\
		&&&\theta_n\in [-\pi, \pi),n=1,\ldots,{N_I}.\label{constraint3}
		\end{align}
	\end{subequations}

The problems \begin{small}$(P1)$\end{small} and \begin{small}$(P2)$\end{small} are non-convex due to the non-concave objective functions with respect to \begin{small}$\boldsymbol\Sigma_s, \boldsymbol\Sigma_z$\end{small} and \begin{small}$\theta_n$\end{small} and the expectations in the objective functions. Unfortunately, there is no standard method for solving them. 

Prior to solving the problems \begin{small}$(P1)$\end{small} and \begin{small}$(P2)$\end{small}, we present the feasibility. For notational simplicity, we define three sets \begin{small}$\mathcal{X}_1\triangleq\left\{\mathbf{\Sigma}_s :\Tr\left(\mathbf{\mathbf{\Sigma}_s}\right)\leq 1, \mathbf{\Sigma}_s \succeq 0\right\}$\end{small}, \begin{small}$\mathcal{X}_2\triangleq\{(\mathbf{\Sigma}_s, \mathbf{\Sigma}_z) :\Tr\left(\mathbf{\mathbf{\Sigma}_s}+\mathbf{\mathbf{\Sigma}_z}\right)\leq 1, \mathbf{\Sigma}_s \succeq 0, \mathbf{\Sigma}_z \succeq 0\}$\end{small}, and \begin{small}$\mathcal{Y}\triangleq\left\{\mathbf{\Theta}:\theta_n\in [-\pi, \pi),n=1,\ldots,{N_I}\right\}$\end{small}. There exists a feasible point that \begin{small}$\mathbf\Sigma_{s}=\frac{1}{N_t}\mathbf I$\end{small} and \begin{small}$\theta_n=0,n=1,\ldots,{N_I}$\end{small} satisfies the constraint \begin{small}$\mathcal{X}_1\cap\mathcal{Y}$\end{small}, hence the problem \begin{small}$(P1)$\end{small} is feasible. Similarly, the point that \begin{small}$\mathbf\Sigma_{s}=\frac{1}{N_t}\mathbf I, \mathbf\Sigma_{z}=\mathbf 0$\end{small} and \begin{small}$\theta_n=0,n=1,\ldots,{N_I}$\end{small} satisfies the constraint \begin{small}$\mathcal{X}_2\cap\mathcal{Y}$\end{small}, which verifies the feasible of the problem \begin{small}$(P2)$\end{small} \cite[\S 4.1.1]{Convex Optimization}. In the sequel, we will develop some iterative algorithms to solve \begin{small}$(P1)$\end{small} and \begin{small}$(P2)$\end{small} efficiently.

\section{Two Frameworks for Secrecy Rates Maximization}
Overall, both of our proposed frameworks decompose the original problem into two kinds of problems, namely optimizing phase shifts with fixed AN-aided beamforming and optimizing AN-aided beamforming with fixed phase shifts. In this section, we first study the effect of phase shifts on legitimate receiver and eavesdropper with the statistics of  channel coefficients. Then, we give the iterative structure of our frameworks and study their performance. Note that we mainly concentrate on the realization of \begin{small}$C_3(\mathbf{\Sigma}_s, \mathbf{\Sigma}_z, \mathbf{\Theta})$\end{small} in the following algorithms in this section, since \begin{small}$C_3(\mathbf{\Sigma}_s, \mathbf{\Sigma}_z, \mathbf{\Theta})$\end{small} contains all variables and \begin{small}$C_1(\mathbf{\Sigma}_s, \mathbf{\Theta}), C_2(\mathbf{\Sigma}_s, \mathbf{\Theta}), C_4(\mathbf{\Sigma}_s, \mathbf{\Sigma}_z, \mathbf{\Theta})$\end{small} have similar optimization structures with \begin{small}$C_3(\mathbf{\Sigma}_s, \mathbf{\Sigma}_z, \mathbf{\Theta})$\end{small} in our frameworks.
\subsection{Phase Shifts Optimization}
\begin{Proposition}
\label{theta effect}
If the channels between the transmitter and legitimate and eavesdropper are i.i.d. Gaussian fading, phase shifts have no contribution to the expectation terms.
\end{Proposition}

\begin{proof}
	Since each element in \begin{small}$\mathbf{h}_r$\end{small} and \begin{small}$\mathbf{H}_e$\end{small} is distributed as \begin{small}$\mathcal{CN}(0,1)$\end{small}, the distributions of \begin{small}$\mathbf{h}_r$\end{small} and \begin{small}$\mathbf{H}_e$\end{small} could be expressed as \begin{small}$\mathbf{h}_r\sim \mathcal{CN}\left(\mathbf 0, \mathbf I_{N_I}\right)$\end{small} and \begin{small}$\Vector\left(\mathbf{H}_e\right)\sim \mathcal{CN}\left(\mathbf 0, \mathbf I_{N_I N_e}\right)$\end{small}, respectively. For any unitary matrices \begin{small}$\mathbf{\Theta}$\end{small} and \begin{small}$\mathbf I_{N_e}\otimes\mathbf{\Theta}$\end{small}, we have
	\begin{equation}\small
	\begin{aligned}
	&\mathbf{\Theta}\mathbf{h}_r\sim \mathcal{CN}\left(\mathbf 0, \mathbf I_{N_I}\right)\\
	&\Vector\left(\mathbf{\Theta}\mathbf{H}_e\right)=\left(\mathbf I_{N_e}\otimes\mathbf{\Theta}\right)\Vector\left(\mathbf{H}_e\right)\sim \mathcal{CN}\left(\mathbf 0, \mathbf I_{N_I N_e}\right).
	\end{aligned}
	\end{equation}
	Hence, \begin{small}$\mathbf{\Theta}\mathbf{h}_r$\end{small} and \begin{small}$\mathbf{\Theta}\mathbf{H}_e$\end{small} have the same distributions as \begin{small}$\mathbf{h}_r$\end{small} and \begin{small}$\mathbf{H}_e$\end{small}, respectively \cite[A.26]{FOWC}. The same distribution implies that phase shifts have no contribution to the expectation terms. Then we have
	\begin{equation}\small
	\begin{aligned}
	&\Expectation_{\mathbf{h}_r} \left[\log\left(1+\rho_r\mathbf{h}_r^H \mathbf{\Theta} \mathbf{G}^H\mathbf{\Sigma}_s \mathbf{G} \mathbf{\Theta}^H \mathbf{h}_r\right)\right]\\
	&=\Expectation_{\mathbf{h}_r} \left[\log\left(1+\rho_r\mathbf{h}_r^H \mathbf{G}^H\mathbf{\Sigma}_s \mathbf{G} \mathbf{h}_r\right)\right]\\
	&\Expectation_{\mathbf{H}_e} \left \{\log[\det ( \mathbf{I}+\rho_e\mathbf{H}_e^H \mathbf{\Theta} \mathbf{G}^H\mathbf{\Sigma}_s \mathbf{G} \mathbf{\Theta}^H \mathbf{H}_e)]\right\}\\
	& = \Expectation_{\mathbf{H}_e} \left \{\log[\det ( \mathbf{I}+\rho_e\mathbf{H}_e^H \mathbf{G}^H\mathbf{\Sigma}_s \mathbf{G} \mathbf{H}_e)]\right\}.
	\end{aligned}
	\end{equation}
\end{proof}

For given \begin{small}$\mathbf\Sigma_s$\end{small} and \begin{small}$\mathbf\Sigma_z$\end{small}, we denote \begin{small}$\mathbf{Y}_1 = \rho_r\diag(\mathbf{h}_r^H) \mathbf{G}^H\mathbf{\Sigma}_s \\\mathbf{G} \diag(\mathbf{h}_r)$\end{small}, \begin{small}$\mathbf{Y}_2 = \frac{1}{{N_I}}\mathbf{I}+\rho_r\diag\left(\mathbf{h}_r\right) \mathbf{G}^H\mathbf{\Sigma}_z \mathbf{G} \diag\left(\mathbf{h}_r\right)$\end{small} and \begin{small}$\mathbf{v} \triangleq \diag(\mathbf{\Theta}^H)$\end{small}. Then we have the following problem:
\begin{equation*}\small
\label{SolveTheta1}
\begin{aligned}
(P3)\,\, \underset{\mathbf{v}}{\max}\quad \frac{\mathbf{v}^H\mathbf{Y}_1\mathbf{v}}{\mathbf{v}^H\mathbf{Y}_2\mathbf{v}}\quad\quad\quad\quad\mathop{\text{s.t.}}\quad\text{(\ref{constraint3})}
\end{aligned}
\end{equation*}
This problem belongs to fractional programming, and we consider the corresponding parametric program: 
\begin{equation*}\small
\label{(P3.1)}
(P3.1)\,\,\underset{\mathbf{v}}{\min}\quad\mathbf{v}^H\left(\mathbf{Y}_2-\mu\mathbf{Y}_1\right)\mathbf{v}\quad\quad\quad\quad\mathop{\text{s.t.}}\quad\text{(\ref{constraint3})}
\end{equation*}
where \begin{small}$\mu>0$\end{small} is an introduced parameter. Many works have emerged to solve the optimization problem efficiently with unit modulus constraint such as MM \cite{Secure_IRS1}-\cite{Secure_IRS2}, manifold optimization \cite{Secure_IRS4}, SDR \cite{Secure_IRS3}\cite{Secure_IRS5} and DRL \cite{C. Yuen3}. Here, we adopt MM algorithm to solve \begin{small}$(P3.1)$\end{small} due to its closed-form solution at each iteration. At iteration \begin{small}$n$\end{small}, for any feasible point \begin{small}$\mathbf{v}^n$\end{small}, \begin{small}$\mathbf{v}^H\left(\mathbf{Y}_2-\mu\mathbf{Y}_1\right)\mathbf{v}$\end{small} can be upper bounded by \begin{small}$\lambda_{\text{max}}\left(\mathbf{\Phi}\right)\Vert \mathbf{v} \Vert^2-2 \mathbb{R}\left(\mathbf{v}^H\boldsymbol{\beta}\right)+c$\end{small}, where \begin{small}$\mathbf{\Phi}=\mathbf{Y}_2-\mu\mathbf{Y}_1, \boldsymbol{\beta} = \left(\mathbf{\Phi}-\lambda_{\text{max}}\left(\mathbf{\Phi}\right)\mathbf{I}\right)\mathbf{v}^{n}$\end{small}, and \begin{small}$c = (\mathbf{v}^{n})^H\left(\lambda_{\text{max}}\left(\mathbf{\Phi}\right)\mathbf{I}_{{N_I}}-\mathbf{\Phi}\right)\mathbf{v}^{n}$\end{small}. Then we have the following problem:
\begin{equation*}\small
\label{(P3.2)}
\begin{aligned}
(P3.2)\,\,\underset{\mathbf{v}}{\min}&\quad \lambda_{\text{max}}\left(\mathbf{\Phi}\right)\Vert \mathbf{v} \Vert^2-2 \mathbb{R}\left(\mathbf{v}^H\boldsymbol{\beta}\right)+c
\end{aligned}
\end{equation*}
Moreover, \begin{small}$\mathbb{R}\left(\mathbf{v}^H\boldsymbol{\beta}\right)$\end{small} is maximized when the phases of \begin{small}$v_i$\end{small} and \begin{small}$\beta_i$\end{small} are equal, where \begin{small}$v_i$\end{small} and \begin{small}$\beta_i$\end{small} are the \begin{small}$i$\end{small}th entry of \begin{small}$\mathbf{v}$\end{small} and \begin{small}$\boldsymbol{\beta}$\end{small}, respectively. Therefore, the optimal solution at iteration \begin{small}$n$\end{small} is
\begin{equation}\small
\label{update_theta1}
\mathbf{v}^{n+1}=\left[e^{j\arg(\beta_1)},\ldots,e^{j\arg(\beta_{{N_I}})}\right]^T .
\end{equation}
More details are provided in Algorithm 1.

\begin{algorithm}
	\caption{MM Algorithm for Phase Shifts Optimization}
	{\bf Input:} \begin{small}$\mathbf{\Sigma}_s$\end{small}, \begin{small}$\mathbf{\Sigma}_z$\end{small} and the initial point \begin{small}$\mathbf{v}^0$\end{small}.
	\begin{algorithmic}[1]
		\State Compute \begin{small}$\mu_{\text{max}}=\lambda_{\text{max}}(\mathbf{Y_1})/\lambda_{\text{min}}(\mathbf{Y_2})$\end{small} and \begin{small}$\mu_{\text{min}}=\lambda_{\text{min}}(\mathbf{Y_1})/\lambda_{\text{min}}(\mathbf{Y_2})$\end{small};
		\Repeat
			\State Set \begin{small}$\mu=(\mu_{\text{max}}+\mu_{\text{min}})/2$\end{small}, and \begin{small}$n=0$\end{small};
			\Repeat
			\State Update \begin{small}$\boldsymbol{\beta} = \left(\mathbf{\Phi}-\lambda_{\text{max}}\left(\mathbf{\Phi}\right)\mathbf{I}\right)\mathbf{v}^{n}$\end{small} where \begin{small}$\mathbf{\Phi}=\mathbf{Y}_2-\mu\mathbf{Y}_1$\end{small};
			\State Optimize \begin{small}$\mathbf{v}^{n+1}$\end{small} according to the closed-form (\ref{update_theta1});
			\State \begin{small}$n\leftarrow n+1$\end{small};
			\Until{Convergence}
			\If{\begin{small}$(\mathbf{v}^{n})^H(\mathbf{Y}_1-\mu\mathbf{Y}_2)\mathbf{v}^{n}> 0$\end{small}} { set \begin{small}$\mu_{\text{min}}=\mu$\end{small}}\ElsIf{\begin{small}$(\mathbf{v}^{n})^H(\mathbf{Y}_1-\mu\mathbf{Y}_2)\mathbf{v}^{n}< 0$\end{small}} { set \begin{small}$\mu_{\text{max}}=\mu$\end{small}}
			\EndIf
		\Until{\begin{small}$(\mathbf{v}^{n})^H(\mathbf{Y}_1-\mu\mathbf{Y}_2)\mathbf{v}^{n} = 0$\end{small}}
	\end{algorithmic}
	{\bf Output:} \begin{small}$\mathbf{v}^*=\mathbf{v}^n$\end{small}.
\end{algorithm}

\subsection{SAA-based Algorithm}
The basic idea of the SAA-based algorithm is to generate some independent random samples and approximate the expectation function by the corresponding sample average function. Then, the result of the sample average optimization problem is considered as an approximate solution to the original problem \cite{SAA}. We define \begin{small}$C_3(\mathbf{\Sigma}_s, \mathbf{\Sigma}_z, \mathbf{\Theta}, \mathbf{H}_e)\triangleq\log\left(1+\frac{\rho_r\mathbf{h}_r^H \mathbf{\Theta} \mathbf{G}^H\mathbf{\Sigma}_s \mathbf{G} \mathbf{\Theta}^H \mathbf{h}_r}{1+\rho_r\mathbf{h}_r^H \mathbf{\Theta} \mathbf{G}^H\mathbf{\Sigma}_z \mathbf{G} \mathbf{\Theta}^H \mathbf{h}_r}\right)-\log\left[\det \left( \mathbf{I}+\rho_e\mathbf{H}_e^H \mathbf{\Theta} \mathbf{G}^H\left(\mathbf{\Sigma}_s+\mathbf{\Sigma}_z\right) \mathbf{G} \mathbf{\Theta}^H \mathbf{H}_e\right)\right]+\log\left[\det \left( \mathbf{I}+\rho_e\mathbf{H}_e^H \mathbf{\Theta} \mathbf{G}^H\mathbf{\Sigma}_z \mathbf{G} \mathbf{\Theta}^H \mathbf{H}_e\right)\right]$
\end{small}, the expectation of which is \begin{small}$\Expectation_{\mathbf{H}_e} \left[C_3(\mathbf{\Sigma}_s, \mathbf{\Sigma}_z, \mathbf{\Theta}, \mathbf{H}_e)\right]=C_3(\mathbf{\Sigma}_s, \mathbf{\Sigma}_z, \mathbf{\Theta})$\end{small} in (\ref{rate3}). Lemma \ref{chebyshev} shows that the gap between the approximate function and the original function becomes smaller as the sample size becomes larger due to the assumption of the i.i.d. channel \begin{small}$\mathbf{H}_e$\end{small}.
	\begin{Lemma}
		\label{chebyshev}
		Suppose that the random matrices \begin{small}$\mathbf{H}_e^1, \mathbf{H}_e^2,...,\mathbf{H}_e^K$\end{small} are independent each other and have the same distribution as \begin{small}$\mathbf{H}_e$\end{small} in (\ref{rate3}), and each \begin{small}$C_3(\mathbf{\Sigma}_s, \mathbf{\Sigma}_z, \mathbf{\Theta}, \mathbf{H}_e^i)$\end{small} is a random variable with mean \begin{small}$s$\end{small} and variance \begin{small}$\delta^2$\end{small}. Let \begin{small}$\overline{C}_3\left(\mathbf{\Sigma}_s, \mathbf{\Sigma}_z, \mathbf{\Theta}, \left(\mathbf{H}_e^i\right)_{i=1}^K\right)= \frac{1}{K}\sum_{i=1}^K C_3(\mathbf{\Sigma}_s, \mathbf{\Sigma}_z, \mathbf{\Theta}, \mathbf{H}_e^i)$\end{small}. Then, Chebyshev's Inequality allows us to write
		\begin{equation}\small
		\begin{aligned}
		&Pr\left[\bigg\vert \overline{C}_3\left(\mathbf{\Sigma}_s, \mathbf{\Sigma}_z, \mathbf{\Theta}, \left(\mathbf{H}_e^i\right)_{i=1}^K\right)-s\bigg\vert\geq\epsilon\right]\leq\frac{\delta^2}{K\epsilon^2}
		\end{aligned}
		\end{equation}
		for any fixed \begin{small}$\epsilon$\end{small}.
	\end{Lemma}
For given \begin{small}$\mathbf{\Theta}$\end{small}, we define \begin{small}
	$f_r\left(\mathbf{X}\right) \triangleq \log(1+\rho_r\mathbf{h}_r^H \mathbf{\Theta} \mathbf{G}^H\mathbf{X} \mathbf{G} \mathbf{\Theta}^H\\ \mathbf{h}_r)$
\end{small}, and \begin{small}$f_e^i\left(\mathbf{X}\right) \triangleq \log[\det ( \mathbf{I}+\rho_e\mathbf{H}_e^{iH} \mathbf{G}^H\mathbf{X} \mathbf{G} \mathbf{H}_e^i)]$\end{small}. Based on Lemma 1, \begin{small}$C_3\left(\mathbf{\Sigma}_s,\mathbf{\Sigma}_z, \mathbf{\Theta}\right)$\end{small} is approximated as
	\begin{equation}\small
	\label{approximation c3}
	\begin{aligned}
	C_3\left(\mathbf{\Sigma}_s,\mathbf{\Sigma}_z, \mathbf{\Theta}\right) &\approx\overline{C}_{3}\left(\mathbf{\Sigma}_s,\mathbf{\Sigma}_z,\mathbf{\Theta},\left(\mathbf{H}_e^i\right)_{i=1}^K\right)\\
	&=\left[f_r\left(\mathbf{\Sigma}_s+\mathbf{\Sigma}_z\right)-f_r\left(\mathbf{\Sigma}_z\right)\right]\\
	&\quad-\frac{1}{K}\sum_{i=1}^{K}\left[f_e^i\left(\mathbf{\Sigma}_s+\mathbf{\Sigma}_z\right)-f_e^i\left(\mathbf{\Sigma}_z\right)\right].
	\end{aligned}
	\end{equation}
Thus, we have the following problem:
\begin{equation*}\small
	\begin{aligned}
	(P4)\,\,&\mathop{\min}_{\mathbf{\Sigma}_{s},\mathbf{\Sigma}_{z}}&&\!\!\!\!-\overline{C}_{3}\left(\mathbf{\Sigma}_s,\mathbf{\Sigma}_z, \mathbf{\Theta},\left(\mathbf{H}_e^i\right)_{i=1}^K\right)\\
	&\mathop{\text{s.t.}}&&\!\! \text{(\ref{constraint1}),}\,\,\text{(\ref{constraint2})}\nonumber
	\end{aligned}
\end{equation*}
It is obvious that \begin{small}$(P4)$\end{small} is a non-convex programming with the concavity of \begin{small}$f_r\left(\mathbf{\Sigma}_z\right)$\end{small} and \begin{small}$f_e^i(\mathbf{\Sigma}_s+\mathbf{\Sigma}_z)$\end{small}. To facilitate the application of MM \cite{MM_Proposed}, we derive their upper bounds. Besides, considering that a large amount of log-det type functions \begin{small}$-f_e^i\left(\mathbf{\Sigma}_z\right) (i=1,\cdots,K)$\end{small} seriously influence the complexity in each iteration, we use descent lemma \cite[Lemma 12]{MM_Proposed} to get the upper bound due to the fact that the gradients of them are Lipschitz continuous, which has been verified in Appendix \ref{Estimation Lipschitz constant}. At iteration \begin{small}$t$\end{small}, \begin{small}$f_r\left(\mathbf{\Sigma}_z\right)$\end{small}, \begin{small}$f_e^i\left(\mathbf{\Sigma}_s+\mathbf{\Sigma}_z\right)$\end{small} and \begin{small}$-f_e^i\left(\mathbf{\Sigma}_z\right)$\end{small} satisfy the following inequalities for any feasible points \begin{small}$\mathbf\Sigma_s^{t}$\end{small} and \begin{small}$\mathbf\Sigma_z^{t}$\end{small}:
\begin{equation}\small
\begin{aligned}
f_r\left(\mathbf{\Sigma}_z\right) &\leq f_r\left(\mathbf{\Sigma}_z^t\right)+\Tr\left[\left(\nabla f_r\left(\mathbf{\Sigma}_z^t\right)\right)^T\left(\mathbf\Sigma_z-\mathbf\Sigma_z^t\right)\right]\\
f_e^i\left(\mathbf{\Sigma}_s+\mathbf{\Sigma}_z\right) &\leq f_e^i\left(\mathbf{\Sigma}_s^t+\mathbf{\Sigma}_z^t\right)+\Tr\bigg[\left(\nabla f_e^i\left(\mathbf{\Sigma}_s^t+\mathbf{\Sigma}_z^t\right)\right)^T\\
&\qquad\left(\mathbf{\Sigma}_s+\mathbf{\Sigma}_z-\mathbf{\Sigma}_s^t-\mathbf{\Sigma}_z^t\right)\bigg]\\
-f_e^i\left(\mathbf{\Sigma}_z\right)&\leq-f_e^i\left(\mathbf{\Sigma}_z^t\right)-\Tr\left[\left(\nabla f_e^i\left(\mathbf{\Sigma}_z^t\right)\right)^T \left(\mathbf{\Sigma}_z-\mathbf{\Sigma}_z^t\right)\right]\\
&\qquad+L_i\|\mathbf{\Sigma}_z-\mathbf{\Sigma}_z^t\|^2
\end{aligned}
\end{equation}
with the derivative functions \begin{small}$\nabla f_r\left(\mathbf{X}\right)$\end{small} and \begin{small}$\nabla f_e^i\left(\mathbf{X}\right)$\end{small} shown as follows
\begin{equation}\small
\label{tidu_SAA}
\begin{aligned}
&\nabla f_r\left(\mathbf{X}\right)=\left(\frac{\rho_r \mathbf{G} \mathbf{\Theta}^H \mathbf{h}_r\mathbf{h}_r^H \mathbf{\Theta} \mathbf{G}^H}{1+\rho_r\mathbf{h}_r^H \mathbf{\Theta} \mathbf{G}^H\mathbf{X} \mathbf{G} \mathbf{\Theta}^H \mathbf{h}_r}\right)^T\\
&\nabla f_e^i\left(\mathbf{X}\right) =\left( \rho_e\mathbf{G}\mathbf{H}_e^i\left(\mathbf{I}+\rho_e\mathbf{H}_e^{iH} \mathbf{G}^H\mathbf{X} \mathbf{G}\mathbf{H}_e^i \right)^{-1}\mathbf{H}_e^{iH} \mathbf{G}^H\right)^T\\
\end{aligned}
\end{equation}
where \begin{small}$L_i, i=1,2,...,K$\end{small} are the Lipschitz constants of \begin{small}$\nabla f_e^i\left(\mathbf{\Sigma}_z\right)$\end{small}. Then, we minimize the surrogate function as follows:
\begin{equation}\small
\begin{aligned}
(P4.1)\,\,&\mathop{\min}_{\mathbf{\Sigma}_{s},\mathbf{\Sigma}_{z}} &&\!\!\!\! -f_r\left(\mathbf{\Sigma}_s+\mathbf{\Sigma}_z\right)+\Tr\left(\mathbf{A}_s^T\mathbf{\Sigma}_s\right)+\Tr\left(\mathbf{A}_z^T\mathbf{\Sigma}_z\right)\\
&&&+\frac{1}{K}\sum_{i=1}^K L_i\|\mathbf{\Sigma}_z-\mathbf{\Sigma}_z^t\|^2\\
&\mathop{\text{s.t.}}&&\!\! \text{(\ref{constraint1}),}\,\,\text{(\ref{constraint2})}\nonumber
\end{aligned}
\end{equation}
where \begin{small}$\mathbf{A}_s=\frac{1}{K}\sum_{i=1}^K\nabla f_e^i\left(\mathbf{\Sigma}_s^t+\mathbf{\Sigma}_z^t\right)$\end{small}, and \begin{small}$\mathbf{A}_z=\nabla f_r\left(\mathbf{\Sigma}_z^t\right)+\frac{1}{K}\sum_{i=1}^K\left[\nabla f_e^i\left(\mathbf{\Sigma}_s^t+\mathbf{\Sigma}_z^t\right)-\nabla f_e^i\left(\mathbf{\Sigma}_z^t\right)\right]$\end{small}. We have proved that for any \begin{small}$L_i\geq \left(\rho_e N_e \Vert \mathbf{G}\mathbf{H}_e^i\mathbf{H}_e^{iH} \mathbf{G}^H\Vert\right)^2$\end{small}, \begin{small}$\Vert\nabla f_e^i\left(\mathbf{X}\right)-\nabla f_e^i\left(\mathbf{Y}\right)\Vert\leq L_i\Vert\mathbf{X}-\mathbf{Y}\Vert$\end{small} holds in Appendix \ref{Estimation Lipschitz constant}. But if \begin{small}$\left(\rho_e N_e \Vert \mathbf{G}\mathbf{H}_e^i\mathbf{H}_e^{iH} \mathbf{G}^H\Vert\right)^2$\end{small} is too large, then it will cause our algorithm to converge very slowly, which is similar to the effect of a too small step-size on the classical gradient descent algorithm. Hence, by replacing \begin{small}$\frac{1}{K}\sum_{i=1}^K L_i$\end{small} with an adaptive Lipschitz constant \begin{small}$L_t$\end{small} at iteration \begin{small}$t$\end{small}, we have
\begin{equation}\small
\begin{aligned}
(P4.2)\,\,&\mathop{\min}_{\mathbf{\Sigma}_{s},\mathbf{\Sigma}_{z}} &&\!\!\!\! -f_r\left(\mathbf{\Sigma}_s+\mathbf{\Sigma}_z\right)+\Tr\left(\mathbf{A}_s^T\mathbf{\Sigma}_s\right)+\Tr\left(\mathbf{A}_z^T\mathbf{\Sigma}_z\right)\\
&&&+L_t\|\mathbf{\Sigma}_z-\mathbf{\Sigma}_z^t\|^2\\
&\mathop{\text{s.t.}}&&\!\! \text{(\ref{constraint1}),}\,\,\text{(\ref{constraint2})}\nonumber
\end{aligned}
\end{equation}
As \begin{small}$(P4.2)$\end{small} is a convex semidefinite program problem, it can be solved by existing convex optimization solvers such as CVX \cite{CVX}. More details are provided in Algorithm 2.

\begin{algorithm}
	\caption{SAA-based Algorithm}
	{\bf Input:}
	\begin{small}$\mathbf{\Sigma}_s^0, \mathbf{\Sigma}_z^0,  \mathbf{\Theta}^0, t=0, L_0$\end{small};
	\begin{algorithmic}[1]
		\State The independent random matrices (\begin{small}$\mathbf{H}_e^1, \mathbf{H}_e^2,\ldots,\mathbf{H}_e^K$\end{small}) are generated;
		\State Compute \begin{small}$\overline{C}_{3}\left(\mathbf{\Sigma}_s^t,\mathbf{\Sigma}_z^t,\mathbf{\Theta}^t,\left(\mathbf{H}_e^i\right)_{i=1}^K\right)$\end{small};
		\For{\begin{small}$t=0,1,2,...$\end{small}}
		\State \begin{small}$L_{t+1}=L_t/4$\end{small};
		\Repeat
		\State \begin{small}$L_{t+1}=2L_{t+1}$\end{small};
		\State Update \begin{small}$(\mathbf{\Sigma}_s^{t+1}, \mathbf{\Sigma}_z^{t+1})$\end{small} for given \begin{small}$(\mathbf{\Sigma}_s^{t}, \mathbf{\Sigma}_z^{t}, \mathbf{\Theta}^{t})$\end{small} according to a series of problems \begin{small}$(P4.2)$\end{small};
		\Until {\begin{small}$\overline{C}_{3}\left(\mathbf{\Sigma}_s^{t+1},\mathbf{\Sigma}_z^{t+1},\mathbf{\Theta}^{t},\left(\mathbf{H}_e^i\right)_{i=1}^K\right)\geq\overline{C}_{3}\left(\mathbf{\Sigma}_s^t,\mathbf{\Sigma}_z^t,\mathbf{\Theta}^t,\left(\mathbf{H}_e^i\right)_{i=1}^K\right)$\end{small};}
		\State Update \begin{small}$\mathbf\Theta^{t+1}$\end{small} for given \begin{small}$(\mathbf{\Sigma}_s^{t+1}, \mathbf{\Sigma}_z^{t+1}, \mathbf{\Theta}^{t})$\end{small} according to Algorithm 1;
		\State Compute \begin{small}$\overline{C}_{3}\left(\mathbf{\Sigma}_s^{t+1},\mathbf{\Sigma}_z^{t+1},\mathbf{\Theta}^{t+1},\left(\mathbf{H}_e^i\right)_{i=1}^K\right)$\end{small};
		\EndFor
	\end{algorithmic}
	{\bf Output:}\begin{small}$\mathbf{\Sigma}_s=\mathbf{\Sigma}_s^{t+1}, \mathbf{\Sigma}_z=\mathbf{\Sigma}_z^{t+1}, \mathbf{\Theta}=\mathbf{\Theta}^{t+1}$\end{small}.
\end{algorithm}

A key feature of Algorithm 2 is the adaptive line search: it always tries to use a smaller Lipschitz constant at the beginning of each iteration by setting \begin{small}$L_{t+1}=L_t/4$\end{small}. The line search procedure starts with an estimated Lipschitz constant \begin{small}$L_t$\end{small}, and increases its value by a factor of $2$ until the construction rule of MM is satisfied \cite[Eqn. (24)]{MM_Proposed}. In other words, the stopping criteria that \begin{small}$\overline{C}_{3}\left(\mathbf{\Sigma}_s^{t+1},\mathbf{\Sigma}_z^{t+1},\mathbf{\Theta}^{t},\left(\mathbf{H}_e^i\right)_{i=1}^K\right)\geq\overline{C}_{3}\left(\mathbf{\Sigma}_s^t,\mathbf{\Sigma}_z^t,\mathbf{\Theta}^t,\left(\mathbf{H}_e^i\right)_{i=1}^K\right)$\end{small} for line search is also satisfied \cite{Adaptive_Lipschitz1}\cite{Adaptive_Lipschitz2}.

\subsection{Hybrid SPG-CP Algorithm}
In the hybrid SPG-CP algorithm, we take the SPG-based algorithm and MM algorithm to optimize \begin{small}$\left(\mathbf{\Sigma}_s, \mathbf{\Sigma}_z\right)$\end{small} and \begin{small}$\mathbf{\Theta}$\end{small}, respectively. 

Here we first introduce some conclusions about the generalized projection associated with the proximal operator. Let \begin{small}$P_{\mathcal{X}_2}\left[(\mathbf\Sigma_{s},\mathbf\Sigma_{z}),(\mathbf{H}_s,\mathbf{H}_z), r\right]$\end{small} denote the projection of \begin{small}$(\mathbf{H}_s,\mathbf{H}_z)$\end{small} onto the set \begin{small}$\mathcal{X}_2$\end{small} at the point \begin{small}$(\mathbf\Sigma_{s},\mathbf\Sigma_{z})\in \mathcal{X}_2$\end{small} with the step size \begin{small}$r>0$\end{small}. We define
\begin{equation}\small
\label{Projection definition}
P_{\mathcal{X}_2}\left[(\mathbf\Sigma_{s},\mathbf\Sigma_{z}),(\mathbf{H}_s,\mathbf{H}_z), r\right]\triangleq\frac{1}{r}\left[(\mathbf\Sigma_{s},\mathbf\Sigma_{z})-(\mathbf\Sigma_{s}^o,\mathbf\Sigma_{z}^o)\right]
\end{equation}
where the proximal operator \begin{small}$(\mathbf\Sigma_{s}^o,\mathbf\Sigma_{z}^o)$\end{small} is defined as \begin{small}$(\mathbf\Sigma_{s}^o,\mathbf\Sigma_{z}^o)\triangleq\mathop{\arg \min}\limits _{(\mathbf X_s,\mathbf X_z) \in \mathcal{X}_2}\{2\langle \mathbf{H}_s, \mathbf X_s\rangle+2\langle \mathbf{H}_z, \mathbf X_z\rangle+\frac{1}{r} \Vert \mathbf X_s-\mathbf\Sigma_{s} \Vert^2+\frac{1}{r} \Vert \mathbf X_z-\mathbf\Sigma_{z} \Vert^2\}$\end{small}; Then \begin{small}$(\mathbf\Sigma_{s}^o,\mathbf\Sigma_{z}^o)$\end{small} satisfies \begin{small}$\langle\mathbf X_s-\mathbf\Sigma_{s}^o, \mathbf{H}_s+\frac{1}{r}\left(\mathbf\Sigma_{s}^o-\mathbf\Sigma_{s}\right)\rangle\geq0$\end{small} and \begin{small}$\langle\mathbf X_z-\mathbf\Sigma_{z}^o, \mathbf{H}_z+\frac{1}{r}\left(\mathbf\Sigma_{z}^o-\mathbf\Sigma_{z}\right)\rangle\geq0$\end{small}. If \begin{small}$P_{\mathcal{X}_2}\left[(\mathbf\Sigma_{s},\mathbf\Sigma_{z}),(\mathbf{H}_s,\mathbf{H}_z), r\right]=\frac{1}{r}\left[(\mathbf\Sigma_{s},\mathbf\Sigma_{z})-(\mathbf\Sigma_{s}^o,\mathbf\Sigma_{z}^o)\right]=\mathbf{0}$\end{small}, i.e. \begin{small}$\mathbf\Sigma_{s}=\mathbf\Sigma_{s}^o, \mathbf\Sigma_{z}=\mathbf\Sigma_{z}^o$\end{small}, we have
\begin{equation}\small
\label{stationarypoint}
\begin{aligned}
\langle\mathbf X_s-\mathbf\Sigma_{s}^o, \mathbf{H}_s+\frac{1}{r}\left(\mathbf\Sigma_{s}^o-\mathbf\Sigma_{s}\right)\rangle&=\langle\mathbf X_s-\mathbf\Sigma_{s}, \mathbf{H}_s\rangle\geq 0 \\
\langle\mathbf X_z-\mathbf\Sigma_{z}^o, \mathbf{H}_z+\frac{1}{r}\left(\mathbf\Sigma_{z}^o-\mathbf\Sigma_{z}\right)\rangle&=\langle\mathbf X_z-\mathbf\Sigma_{z}, \mathbf{H}_z\rangle\geq 0\\
& \forall (\mathbf X_s,\mathbf X_z) \in \mathcal{X}_2
\end{aligned}
\end{equation}
\begin{Lemma}
	\label{Projection1}
	For any \begin{small}$(\mathbf\Sigma_{s},\mathbf\Sigma_{z})\in \mathcal{X}_2$\end{small}, \begin{small}$(\mathbf{H}_s,\mathbf{H}_z)$\end{small}, and \begin{small}$r>0$\end{small}, we have
	\begin{equation}\small
	\begin{aligned}
	&\langle (\mathbf{H}_s,\mathbf{H}_z), P_{\mathcal{X}_2}\left[(\mathbf\Sigma_{s},\mathbf\Sigma_{z}),(\mathbf{H}_s,\mathbf{H}_z), r\right]\rangle\\
	&\geq \Vert P_{\mathcal{X}_2}\left[(\mathbf\Sigma_{s},\mathbf\Sigma_{z}),(\mathbf{H}_s,\mathbf{H}_z), r\right]\Vert^2.
	\end{aligned}
	\end{equation}
\end{Lemma}
\begin{Lemma}
	\label{Projection3}
	For any \begin{small}$(\mathbf{H}_{s,1},\mathbf{H}_{z,1})$\end{small} and \begin{small}$(\mathbf{H}_{s,2},\mathbf{H}_{z,2})$\end{small}, we have
	\begin{equation}\small
	\begin{aligned}
	\Vert &P_{\mathcal{X}_2}[(\mathbf\Sigma_{s},\mathbf\Sigma_{z}),(\mathbf{H}_{s,1},\mathbf{H}_{z,1}), r]\\
	&-P_{\mathcal{X}_2}\left[(\mathbf\Sigma_{s},\mathbf\Sigma_{z}),(\mathbf{H}_{s,2},\mathbf{H}_{z,2}), r\right]\Vert\\
	&\leq \Vert (\mathbf{H}_{s,1},\mathbf{H}_{z,1})-(\mathbf{H}_{s,2},\mathbf{H}_{z,2})\Vert.
	\end{aligned}
	\end{equation}
\end{Lemma}
Lemma \ref{Projection1}-\ref{Projection3} can be easily proven by following the analysis of \cite[Lemma 1 and Proposition 1]{SGD_REF}.

At iteration $t$, a total of \begin{small}$\lceil t^\alpha \rceil,\alpha>1$\end{small} independent random matrices \begin{small}$\mathcal{H}_e^t\triangleq\left\{\mathbf{H}_{e,t}^i: 1\leq i\leq \lceil t^{\alpha}\rceil \right\}$\end{small} are realized.\footnote{We use the notation \begin{small}$\lceil t^\alpha \rceil$\end{small} to denote the sample size because \begin{small}$\alpha$\end{small} is not necessarily an integer. As we will see in Appendix \ref{Convergence of the Hybrid SPG-CP Algorithm}, \begin{small}$\alpha>1$\end{small} is needed and also enough to prove that the expectation of the projected gradient of \begin{small}$C_3\left(\mathbf{\Sigma}_s,\mathbf{\Sigma}_z, \mathbf{\Theta}\right)$\end{small} approaches $0$ as the number of iterations approaches infinity in Algorithm 3.} Then we define \begin{small}$f_{r,t}\left(\mathbf{X}\right) \triangleq \log\left(1+\rho_r\mathbf{h}_r^H \mathbf{\Theta}^t \mathbf{G}^H\mathbf{X} \mathbf{G} \mathbf{\Theta}^{tH} \mathbf{h}_r\right)$\end{small}, where \begin{small}$\mathbf{\Theta}^t$\end{small} is the solution to \begin{small}$\mathbf{\Theta}$\end{small} in the previous iteration. \begin{small}$f_{e,t}^i \left(\mathbf{X}\right) \triangleq \log\left[\det \left( \mathbf{I}+\rho_e\mathbf{H}_{e,t}^{iH} \mathbf{G}^H\mathbf{X} \mathbf{G} \mathbf{H}_{e,t}^i\right)\right]$\end{small} over \begin{small}$\mathcal{H}_e^t$\end{small}. 
The derivative functions \begin{small}$\nabla f_{r,t}\left(\mathbf{X}\right)$\end{small} and \begin{small}$\nabla f_{e,t}^i\left(\mathbf{X}\right)$\end{small} are shown as follows
	\begin{equation}\small
	\label{tidu_rteti}
	\begin{aligned}
	&\nabla f_{r,t}\left(\mathbf{X}\right)=\left(\frac{\rho_r \mathbf{G} \mathbf{\Theta}^{tH} \mathbf{h}_r\mathbf{h}_r^H \mathbf{\Theta}^t \mathbf{G}^H}{1+\rho_r\mathbf{h}_r^H \mathbf{\Theta}^t \mathbf{G}^H\mathbf{X} \mathbf{G} \mathbf{\Theta}^{tH} \mathbf{h}_r}\right)^T\\
	&\nabla f_{e,t}^i\left(\mathbf{X}\right) =\left( \rho_e\mathbf{G}\mathbf{H}_{e,t}^i\left(\mathbf{I}+\rho_e\mathbf{H}_{e,t}^{iH} \mathbf{G}^H\mathbf{X} \mathbf{G}\mathbf{H}_{e,t}^i \right)^{-1}\mathbf{H}_{e,t}^{iH} \mathbf{G}^H\right)^T.\\
	\end{aligned}
	\end{equation}
Now we solve the following problem for any given \begin{small}$\mathbf{\Sigma}_s^t, \mathbf{\Sigma}_z^t$\end{small}:
\begin{equation}\small
\begin{aligned}
(P5)\,\,\left(\mathbf{\Sigma}_s^{t+1}, \mathbf{\Sigma}_z^{t+1}\right)=&\mathop{\arg\min}_{\mathbf{\Sigma}_{s},\mathbf{\Sigma}_{z}} &&\! 2\langle\mathbf{G}_s^*, \mathbf{\Sigma}_s\rangle+2\langle\mathbf{G}_z^*, \mathbf{\Sigma}_z\rangle\\
&&&\!\!\!\!\!\!\!\!\!\!\!\!+\frac{1}{r}\|\mathbf{\Sigma}_s-\mathbf{\Sigma}_s^t\|^2+\frac{1}{r}\|\mathbf{\Sigma}_z-\mathbf{\Sigma}_z^t\|^2\\
&\mathop{\text{s.t.}}&&\! \text{(\ref{constraint1}),}\,\,\text{(\ref{constraint2})}\nonumber
\end{aligned}
\end{equation}
where \begin{small}$\mathbf{G}_s=-\nabla f_{r,t}(\mathbf{\Sigma}_s^t+\mathbf{\Sigma}_z^t)+\frac{1}{\lceil t^{\alpha}\rceil}\sum_{i=1}^{\lceil t^{\alpha}\rceil}\nabla f_{e,t}^i\left(\mathbf{\Sigma}_s^t+\mathbf{\Sigma}_z^t\right)$\end{small}, and  \begin{small}$\mathbf{G}_z=-\nabla f_{r,t}\left(\mathbf{\Sigma}_s^t+\mathbf{\Sigma}_z^t\right)+\nabla f_{r,t}\left(\mathbf{\Sigma}_z^t\right)+\frac{1}{\lceil t^{\alpha}\rceil}\sum_{i=1}^{\lceil t^{\alpha}\rceil}[\nabla f_{e,t}^i(\mathbf{\Sigma}_s^t+\mathbf{\Sigma}_z^t)-\nabla f_{e,t}^i\left(\mathbf{\Sigma}_z^t\right)]$\end{small}.
The constant \begin{small}$r\in(0,2/L)$\end{small} and $L$ is relevant to the gradient of \begin{small}$C_3(\mathbf{\Sigma}_s, \mathbf{\Sigma}_z, \mathbf{\Theta})$\end{small}, which has been given in Proposition \ref{stateL1}.
A common practice in stochastic optimization is to estimate $L$ by using the stochastic gradients computed at a small number of trial points \cite{SGD_REF2}\cite{Estimate_Lipschitz}.

\begin{Proposition}
\label{stateL1}
	Define \begin{small}$\widehat{\mathbf x} = \left[\Vector{(\mathbf{\Sigma}_s)}^T,\Vector{(\mathbf{\Sigma}_z)}^T\right]^T$\end{small}. The representation of \begin{small}$C_{3}(\mathbf{\Sigma}_s,\mathbf{\Sigma}_z,\mathbf{\Theta})$\end{small} is denoted by \begin{small}$\widehat{C}_3(\widehat{\mathbf x},\mathbf \Theta)=C_{3}\left(\mathbf{\Sigma}_s,\mathbf{\Sigma}_z,\mathbf \Theta\right)$\end{small}. For any \begin{small}$\mathbf\Theta$\end{small}, there exists a constant $L$ such that
	\begin{equation}\small
	\|\nabla\widehat{C}_3(\widehat{\mathbf x}_1,\mathbf \Theta)-\nabla\widehat{C}_3(\widehat{\mathbf x}_2,\mathbf \Theta)\|\leq L\Vert\widehat{\mathbf x}_1-\widehat{\mathbf x}_2\Vert\quad \forall \widehat{\mathbf x}_1,\widehat{\mathbf x}_2
	\end{equation}
	Note that the constant $L$ is irrelevant to \begin{small}$\mathbf \Theta$\end{small}. Moreover,
	\begin{equation}\small
	\begin{aligned} -C_{3}& \left(\mathbf{\Sigma}_s^1,\mathbf{\Sigma}_z^1,\mathbf \Theta\right)\leq-C_{3}\left(\mathbf{\Sigma}_s^2,\mathbf{\Sigma}_z^2,\mathbf \Theta\right)-2\langle \nabla_{\mathbf{\Sigma}_{s}^*}C_{3}\left(\mathbf{\Sigma}_s^2,\mathbf{\Sigma}_z^2,\mathbf \Theta\right),\\
	&\mathbf\Sigma_{s}^1-\mathbf\Sigma_s^2\rangle+L\|\mathbf\Sigma_{s}^1-\mathbf\Sigma_s^2\|^2-2\langle \nabla_{\mathbf{\Sigma}_{z}^*}C_{3}\left(\mathbf{\Sigma}_s^2,\mathbf{\Sigma}_z^2,\mathbf \Theta\right),\\
	&\mathbf\Sigma_{z}^1-\mathbf\Sigma_z^2\rangle+L\|\mathbf\Sigma_{z}^1-\mathbf\Sigma_z^2\|^2
	\end{aligned}
	\end{equation}
\end{Proposition}
\begin{proof}
	See Appendix \ref{Appendix_stateL1}.
\end{proof}

\begin{Proposition}
	The optimal solution of problem \begin{small}$(P5)$\end{small} is given by
	\begin{equation}\small
	\label{Update_sz_SPG-CP}
	\begin{aligned}
	\mathbf{\Sigma}_s^{t+1}=\left[\mathbf{\Sigma}_s^{t}-r\left(\mathbf{G}_s+\overline{\lambda}\mathbf{I}\right)\right]^+\\
	\mathbf{\Sigma}_z^{t+1}=\left[\mathbf{\Sigma}_z^{t}-r\left(\mathbf{G}_z+\overline{\lambda}\mathbf{I}\right)\right]^+
	\end{aligned}
	\end{equation}
	where \begin{small}$\left[\mathbf{X}\right]^+$\end{small} denotes the projection of \begin{small}$\mathbf{X}$\end{small} onto the positive semidefinite cone, and \begin{small}$\overline{\lambda}$\end{small} is the multiplier such that \begin{small}$0\leq\overline{\lambda}\perp\Tr(\mathbf{\Sigma}_s^{t+1}+\mathbf{\Sigma}_z^{t+1})-1\leq 0$\end{small}, which can be found by bisection.
\end{Proposition}
\begin{proof}
	In the following we solve \begin{small}$(P5)$\end{small} via the partial Lagrangian function.
	\begin{equation}\small
	\begin{aligned}
	\mathcal{L}\left(\mathbf{\Sigma}_{s},\mathbf{\Sigma}_{z},\lambda\right)&=2\langle\mathbf{G}_s^*, \mathbf{\Sigma}_s\rangle+2\langle\mathbf{G}_z^*, \mathbf{\Sigma}_z\rangle+\frac{1}{r}\|\mathbf{\Sigma}_s-\mathbf{\Sigma}_s^t\|^2\\
	&+\frac{1}{r}\|\mathbf{\Sigma}_z-\mathbf{\Sigma}_z^t\|^2+\lambda\left(\Tr\left(\mathbf{\mathbf{\Sigma}_s + \mathbf{\Sigma}_z }\right)- 1\right).
	\end{aligned}
	\end{equation}
	The dual function is given by \begin{small}$g(\lambda)=\mathop{\inf}\limits_{\mathbf{\Sigma}_{s},\mathbf{\Sigma}_{z}}\mathcal{L}\left(\mathbf{\Sigma}_{s},\mathbf{\Sigma}_{z},\lambda\right)$\end{small}. Since \begin{small}$\mathcal{L}\left(\mathbf{\Sigma}_{s},\mathbf{\Sigma}_{z},\lambda\right)$\end{small} is a convex function of \begin{small}$\mathbf{\Sigma}_{s},\mathbf{\Sigma}_{z}$\end{small}, we can find the minimizing matrices \begin{small}$\mathbf{\Sigma}_{s},\mathbf{\Sigma}_{z}$\end{small} from the optimality condition
	\begin{equation}\small
	\begin{aligned}
	\nabla_{\mathbf{\Sigma}_{s}^*} \mathcal{L}\left(\mathbf{\Sigma}_{s},\mathbf{\Sigma}_{z},\lambda\right)=\mathbf{G}_s^*+\frac{1}{r}\left(\mathbf{\Sigma}_s-\mathbf{\Sigma}_s^t\right)+\overline{\lambda}\mathbf{I}=\mathbf{0},\\
	\nabla_{\mathbf{\Sigma}_{z}^*} \mathcal{L}\left(\mathbf{\Sigma}_{s},\mathbf{\Sigma}_{z},\lambda\right)=\mathbf{G}_z^*+\frac{1}{r}\left(\mathbf{\Sigma}_z-\mathbf{\Sigma}_z^t\right)+\overline{\lambda}\mathbf{I}=\mathbf{0},
	\end{aligned}
	\end{equation}
	which yields \begin{small}$\mathbf{\Sigma}_s=\mathbf{\Sigma}_s^{t}-r\left(\mathbf{G}_s^*+\overline{\lambda}\mathbf{I}\right)$\end{small} and \begin{small}$\mathbf{\Sigma}_z=\mathbf{\Sigma}_z^{t}-r\left(\mathbf{G}_z^*+\overline{\lambda}\mathbf{I}\right)$\end{small}. Then, \begin{small}$\mathbf{\Sigma}_s$\end{small} and \begin{small}$\mathbf{\Sigma}_z$\end{small} are projected onto the positive semidefinite cone, which leads to the desired (\ref{Update_sz_SPG-CP}).
\end{proof}

\begin{algorithm}
	\caption{Hybrid SPG-CP Algorithm}
	{\bf Input:} Given initial point \begin{small}$(\mathbf{\Sigma}_s^{1}, \mathbf{\Sigma}_z^{1}, \mathbf{\Theta}^{1})$\end{small}, a positive integer $N$, a constant \begin{small}$\alpha>1$\end{small};
	\begin{algorithmic}[1]
		\State Estimate $L$ with a small number of trial points \cite{Estimate_Lipschitz};
		\State Set \begin{small}$r \in (0,2/L)$\end{small};
		\For{\begin{small}$t=1, 2,..., N$\end{small}}
		\State The independent random matrices \begin{small}$\mathbf{H}_{e,t}^1, \mathbf{H}_{e,t}^2, \ldots, \mathbf{H}_{e,t}^{\lceil t^{\alpha}\rceil}$\end{small} are generated;
		\State Update \begin{small}$(\mathbf{\Sigma}_s^{t+1}, \mathbf{\Sigma}_z^{t+1})$\end{small} for given \begin{small}$(\mathbf{\Sigma}_s^{t}, \mathbf{\Sigma}_z^{t}, \mathbf{\Theta}^{t})$\end{small} according to (\ref{Update_sz_SPG-CP});
		\State Update \begin{small}$\mathbf\Theta^{t+1}$\end{small} for given \begin{small}$(\mathbf{\Sigma}_s^{t+1}, \mathbf{\Sigma}_z^{t+1}, \mathbf{\Theta}^{t})$\end{small} according to Algorithm 1;
		\EndFor
	\end{algorithmic}
	{\bf Output:} \begin{small}$(\mathbf{\Sigma}_s^{N+1}, \mathbf{\Sigma}_z^{N+1}, \mathbf{\Theta}^{N+1})$\end{small}
\end{algorithm}

The idea in step 2 is similar with \cite[Algorithm 9]{SSUM}. As \begin{small}$\alpha$\end{small} increases, the sample size \begin{small}$\lceil t^\alpha \rceil$\end{small} increases. Based on Lemma 1, the gradients \begin{small}$\mathbf{G}_s$\end{small} and \begin{small}$\mathbf{G}_z$\end{small} become more accurate in approximating the gradients of \begin{small}$C_3(\mathbf{\Sigma}_s, \mathbf{\Sigma}_z, \mathbf\Theta)$\end{small} with respect to \begin{small}$\mathbf{\Sigma}_s$\end{small} and \begin{small}$\mathbf{\Sigma}_z$\end{small}. However, for a larger \begin{small}$\alpha$\end{small}, computational complexity also increases in each iteration. Therefore, in practice, we need to make a trade-off and select a suitable \begin{small}$\alpha$\end{small} for our algorithm.

\subsection{Convergence Analysis}
For convenience, in what follows we use notations \begin{small}$\mathbf{\Theta}$\end{small} and \begin{small}$\boldsymbol{\theta} \triangleq \left[\theta_1, \theta_2,...,\theta_{{N_I}}\right]^T$\end{small} interchangeably as an argument of a function. Take \begin{small}$\overline{C}_{3}\left(\mathbf{\Sigma}_{x}, \mathbf{\Sigma}_{z}, \boldsymbol{\theta},\left(\mathbf{H}_{c}^{i}\right)_{i=1}^{K}\right)$\end{small} as an example, we have
\begin{equation}\small
\overline{C}_{3}\left(\boldsymbol{\Sigma}_{s}, \boldsymbol{\Sigma}_{z}, \boldsymbol{\theta},\left(\mathbf{H}_{e}^{i}\right)_{i=1}^{K}\right)=\overline{C}_{3}\left(\boldsymbol{\Sigma}_{s}, \boldsymbol{\Sigma}_{z}, \boldsymbol{\Theta},\left(\mathbf{H}_{e}^{i}\right)_{i=1}^{K}\right)
\end{equation}

1. MM Algorithm for Phase Shifts

Define \begin{small}$m_1\left(\boldsymbol{\theta}\right)\triangleq\mathbf{v}^H\mathbf{Y}_1\mathbf{v}$\end{small} and \begin{small}$m_2\left(\boldsymbol{\theta}\right)\triangleq\mathbf{v}^H\mathbf{Y}_2\mathbf{v}$\end{small}. For any given \begin{small}$\mu$\end{small}, steps 4-7 in Algorithm 1 ensure that the sequence \begin{small}$\{\boldsymbol{\theta}^n\}_{n=1}^{\infty}$\end{small} (that is \begin{small}$\{\mathbf{v}^n\}_{n=1}^{\infty}$\end{small}) converges, with the limit point being a local minimizer of the problem \begin{small}$(P3.1)$\end{small} \cite[Proposition 2]{convergence}. Then the stationary point is given by
\begin{equation}\small
\langle \boldsymbol{\theta}-\boldsymbol{\theta}^\infty, \nabla m_1\left(\boldsymbol{\theta}^\infty\right)-\mu\nabla m_2\left(\boldsymbol{\theta}^\infty\right)\rangle\geq 0,\,\,\forall\boldsymbol{\theta}\in\mathcal{Y}.
\end{equation}
In addition, invoking the terminal condition \begin{small}$m_1\left(\boldsymbol{\theta}^\infty\right)-\mu m_2\left(\boldsymbol{\theta}^\infty\right)=0$\end{small}, we have
\begin{equation}\small
\begin{aligned}
\nabla\left(\frac{m_1\left(\boldsymbol{\theta}^\infty\right)}{m_2\left(\boldsymbol{\theta}^\infty\right)}\right)&=\frac{m_2\left(\boldsymbol{\theta}^\infty\right)\nabla m_1\left(\boldsymbol{\theta}^\infty\right)-m_1\left(\boldsymbol{\theta}^\infty\right)\nabla m_2\left(\boldsymbol{\theta}^\infty\right)}{(m_2\left(\boldsymbol{\theta}^\infty\right))^2}\\
&=\frac{m_2\left(\boldsymbol{\theta}^\infty\right)\nabla m_1\left(\boldsymbol{\theta}^\infty\right)-\mu m_2\left(\boldsymbol{\theta}^\infty\right)\nabla m_2\left(\boldsymbol{\theta}^\infty\right)}{(m_2\left(\boldsymbol{\theta}^\infty\right))^2}.\\
\end{aligned}
\end{equation}
Then
\begin{equation}\small
\label{MM1}
\bigg\langle \boldsymbol{\theta}-\boldsymbol{\theta}^\infty, \nabla\left(\frac{m_1\left(\boldsymbol{\theta}^\infty\right)}{m_2\left(\boldsymbol{\theta}^\infty\right)}\right)\bigg\rangle\geq 0,\,\,\forall\boldsymbol{\theta}\in\mathcal{Y}.
\end{equation}
Based on \cite[Eqn. (6)]{MM_Proposed}, (\ref{MM1}) implies that the limit point \begin{small}$\boldsymbol{\theta}^\infty$\end{small} is a stationary point.

2. SAA-based Algorithm

Denote the feasible solutions in the $t$th and $(t+1)$th iterations as \begin{small}$\left(\mathbf{\Sigma}_s^t,\mathbf{\Sigma}_z^t,\boldsymbol{\theta}^t\right)$\end{small} and \begin{small}$(\mathbf{\Sigma}_s^{t+1},\mathbf{\Sigma}_z^{t+1},\boldsymbol{\theta}^{t+1})$\end{small} in Algorithm 2, respectively. It then follows that
\begin{equation}\small
\begin{aligned}
&\overline{C}_{3}\left(\mathbf{\Sigma}_s^{t+1},\mathbf{\Sigma}_z^{t+1},\boldsymbol{\theta}^{t+1}, \left(\mathbf{H}_e^i\right)_{i=1}^K\right)\\
&\geq\overline{C}_{3}\left(\mathbf{\Sigma}_s^{t+1},\mathbf{\Sigma}_z^{t+1},\boldsymbol{\theta}^t,\left(\mathbf{H}_e^i\right)_{i=1}^K\right)\\
&\geq\overline{C}_{3}\left(\mathbf{\Sigma}_s^{t},\mathbf{\Sigma}_z^{t},\boldsymbol{\theta}^t,\left(\mathbf{H}_e^i\right)_{i=1}^K\right).
\end{aligned}
\end{equation}
Hence, we must have
\begin{equation}\small
{\lim_{r_i \to \infty}}\overline{C}_{3}\left(\mathbf{\Sigma}_s^{r_i},\mathbf{\Sigma}_z^{r_i},\boldsymbol{\theta}^{r_i},\left(\mathbf{H}_e^i\right)_{i=1}^K\right)=\overline{C}_{3}\left(\overline{\mathbf{\Sigma}}_s,\overline{\mathbf{\Sigma}}_z,\overline{\boldsymbol{\theta}},\left(\mathbf{H}_e^i\right)_{i=1}^K\right).
\end{equation}
Let \begin{small}$\{\left(\mathbf{\Sigma}_s^{r_1},\mathbf{\Sigma}_z^{r_1},\boldsymbol{\theta}^{r_1}\right),\left(\mathbf{\Sigma}_s^{r_2},\mathbf{\Sigma}_z^{r_2},\boldsymbol{\theta}^{r_2}\right),...,\left(\mathbf{\Sigma}_s^{r_\infty},\mathbf{\Sigma}_z^{r_\infty},\boldsymbol{\theta}^{r_\infty}\right)\}$\end{small} be the subsequence converging to the limit point \begin{small}$\left(\overline{\mathbf{\Sigma}}_s,\overline{\mathbf{\Sigma}}_z,\overline{\boldsymbol{\theta}}\right)$\end{small}. Since \begin{small}$\langle \boldsymbol{\theta}-\boldsymbol{\theta}^{r_i}, -\nabla_{\boldsymbol{\theta}^*}\overline{C}_{3}(\mathbf{\Sigma}_s^{r_i},\mathbf{\Sigma}_z^{r_i},\boldsymbol{\theta}^{r_i},\left(\mathbf{H}_e^i\right)_{i=1}^K)\rangle\geq 0, \forall i$\end{small} holds according to step 9 in Algorithm 2, we take the limit and obtain the inequality  \begin{small}$\langle \boldsymbol{\theta}-\overline{\boldsymbol{\theta}}, -\nabla_{\boldsymbol{\theta}^*}\overline{C}_{3}(\overline{\mathbf{\Sigma}}_s,\overline{\mathbf{\Sigma}}_z,\overline{\boldsymbol{\theta}},\left(\mathbf{H}_e^i\right)_{i=1}^K)\rangle\geq0$\end{small}. On the other hands, due to 
\begin{equation}\small
\begin{aligned}
&\bigg\langle \mathbf{\Sigma}_s-\mathbf{\Sigma}_s^{r_i+1}, -\nabla_{\mathbf{\Sigma}_s^*}\overline{C}_{3}\left(\mathbf{\Sigma}_s^{r_i+1},\mathbf{\Sigma}_z^{r_i+1},\boldsymbol{\theta}^{r_i},\left(\mathbf{H}_e^i\right)_{i=1}^K\right)\bigg\rangle\\
&+\bigg\langle \mathbf{\Sigma}_z-\mathbf{\Sigma}_z^{r_i+1}, -\nabla_{\mathbf{\Sigma}_z^*}\overline{C}_{3}\left(\mathbf{\Sigma}_s^{r_i+1},\mathbf{\Sigma}_z^{r_i+1},\boldsymbol{\theta}^{r_i},\left(\mathbf{H}_e^i\right)_{i=1}^K\right)\bigg\rangle\\
&\geq 0
\end{aligned}
\end{equation} for all possible \begin{small}$\left(\mathbf{\Sigma}_s,\mathbf{\Sigma}_z\right)$\end{small}, we have
\begin{equation}\small
\label{MM2}
\begin{aligned}
&\bigg\langle \mathbf{\Sigma}_s-\overline{\mathbf{\Sigma}}_s, -\nabla_{\mathbf{\Sigma}_s^*}\overline{C}_{3}\left(\overline{\mathbf{\Sigma}}_s,\overline{\mathbf{\Sigma}}_z,\overline{\boldsymbol{\theta}},\left(\mathbf{H}_e^i\right)_{i=1}^K\right)\bigg\rangle\\
&+\bigg\langle \mathbf{\Sigma}_z-\overline{\mathbf{\Sigma}}_z, -\nabla_{\mathbf{\Sigma}_z^*}\overline{C}_{3}\left(\overline{\mathbf{\Sigma}}_s,\overline{\mathbf{\Sigma}}_z,\overline{\boldsymbol{\theta}},\left(\mathbf{H}_e^i\right)_{i=1}^K\right)\bigg\rangle\\
&\geq 0
\end{aligned}
\end{equation}
Based on \cite[Eqn. (6)]{MM_Proposed}, (\ref{MM2}) implies that the limit point \begin{small}$\left(\overline{\mathbf{\Sigma}}_s,\overline{\mathbf{\Sigma}}_z,\overline{\boldsymbol{\theta}}\right)$\end{small} is a stationary point.

3. Hybrid SPG-CP Algorithm

For the performance of Algorithm 3, we have the following proposition.
\begin{Proposition}
	Consider the sequence \begin{small}$\left\{(\mathbf\Sigma_{s}^t,\mathbf\Sigma_{z}^t)\right\}$\end{small} generated by hybrid SPG-CP algorithm. Then the expectation of projected gradient of \begin{small}$(\mathbf{\Sigma}_s^{N}, \mathbf{\Sigma}_z^{N})$\end{small} in Algorithm 3 approaches $0$ as \begin{small}$t=N\to\infty$\end{small}.
\end{Proposition}
\begin{proof}
	See Appendix \ref{Convergence of the Hybrid SPG-CP Algorithm}.
\end{proof}

\subsection{Complexity Analysis}
1. MM Algorithm for Phase Shifts

At the start of Algorithm 1, it is necessary to compute the maximum and the minimum eigenvalues of the matrices \begin{small}$\mathbf Y_1$\end{small} and \begin{small}$\mathbf Y_2$\end{small}, whose complexity is \begin{small}$\mathcal O(N_I^3)$\end{small}. Suppose that the MM algorithm requires \begin{small}$T_1$\end{small} iterations to converge in total. The complexity of each iteration mainly depends on the computation of \begin{small}$\mathbf \beta$\end{small} in step 5 of Algorithm 1 and the corresponding complexity is given by \begin{small}$\mathcal O(N_I^2)$\end{small}. Therefore, the complexity of evaluating \begin{small}$\mathbf v$\end{small} is approximated as \begin{small}$\mathcal O(N_I^3+T_1 N_I^2)$\end{small}.

2. SAA-based Algorithm

The complexity of updating \begin{small}$\mathbf \Sigma_s$\end{small} and \begin{small}$\mathbf \Sigma_z$\end{small} mainly depends on the optimization of problem \begin{small}$(P4.2)$\end{small} and the comparison with the previous solution. Firstly, the complexity of computing \begin{small}$\mathbf{A}_s, \mathbf{A}_z$\end{small} and optimizing \begin{small}$(P4.2)$\end{small} can be asymptotically estimated as \begin{small}$K(N_eN_IN_t+N_eN_t^2+N_e^2N_t+N_e^3)$\end{small} and \begin{small}$\mathcal O(N_t^P),1\leq p\leq 4$\end{small} respectively \cite{convergence}. Secondly, we mainly need to compute the matrix multiplication and determinant in \begin{small}$\overline{C}_3$\end{small}, so the corresponding complexities are \begin{small}$\mathcal O(N_eN_IN_t+N_eN_t^2+N_e^2N_t)$\end{small} and \begin{small}$\mathcal O(N_e^3)$\end{small}, respectively. Suppose that the adaptive algorithm needs \begin{small}$T_2$\end{small} iterations to search for a suitable $L$ and Algorithm 2 requires \begin{small}$T_3$\end{small} iterations to converge. Then the total complexity is approximated as \begin{small}$T_3 (N_I^3+T_1 N_I^2+T_2(N_t^P+N_eN_IN_t+N_eN_t^2+N_e^2N_I+KN_e^2N_I+KN_e^3))$\end{small}.

3. Hybrid SPG-CP Algorithm

Suppose that we need \begin{small}$N_L$\end{small} trial points to estimate $L$, whose complexity is estimated as \begin{small}$\mathcal O (N_L^2(N_eN_IN_t+N_eN_t^2+N_e^2N_t+N_e^3))$\end{small} \cite{Estimate_Lipschitz}. The complexity of updating \begin{small}$\mathbf \Sigma_s$\end{small} and \begin{small}$\mathbf \Sigma_z$\end{small} mainly depends on the search of a suitable \begin{small}$\overline{\lambda}$\end{small} and eigenvalue decomposition. Suppose we need \begin{small}$T_4$\end{small} iterations to get a suitable \begin{small}$\overline{\lambda}$\end{small} by bisection and Algorithm 3 needs $N$ iterations. Then the total complexity is approximated as \begin{small}$\mathcal O (N_L^2(N_eN_IN_t+N_eN_t^2+N_e^2N_I)+N(T_4N_t^3+N_I^3+T_1 N_I^2))$\end{small}.

\section{Analysis on Some Special Cases}
One difficulty in dealing with the maximizing secrecy rate problem is the expectation terms existing in the objective functions. In this section, we utilize some exact result of the expectation and analyze some special cases. To begin with, we identify three important properties: (i) the expectation of log-like function; (ii) the smooth and convex property for a specific optimization problem; (iii) the rank of the optimal \begin{small}$\mathbf{\Sigma}_s$\end{small} for \begin{small}$C_1(\mathbf{\Sigma}_s, \mathbf{\Theta})$\end{small}.
\begin{Lemma}
\label{Int_log}
	Suppose the eigenvalues of \begin{small}$\rho \sigma^{2}\mathbf{Q}$\end{small} are \begin{small}$\{0,0...,0,\tilde{t},\tilde{t},...\tilde{t}\}$\end{small}, where the number of \begin{small}$\tilde{t}'s$\end{small} is \begin{small}$N_1$\end{small}. For \begin{small}$\mathbf{z} \sim \mathcal{CN}(\mathbf{0},\,\sigma^{2}\mathbf{I})$\end{small}, the expectation of \begin{small}$\log\left(1+\mathbf{z}^H\rho\mathbf{Q}\mathbf{z}\right)$\end{small} and its first- and second- order derivatives with respect to \begin{small}$\tilde{t}$\end{small} are
	\begin{subequations}\small
		\label{Nvalues}
		\begin{align}
		\Expectation_{\mathbf{z}} \left\{\log\left(1+\mathbf{z}^H\rho\mathbf{Q}\mathbf{z}\right)\right\} &= \frac{\int_{0}^{\infty}{\log\left(1+\tilde{t}x\right)x^{N_1-1}e^{-x}}dx}{\left(N_1-1\right)!}\nonumber\\
		&\triangleq F_1(\tilde{t},N_1)\\
		\frac{\partial F_1(\tilde{t},N_1)}{\partial \tilde{t}}&= \frac{\int_{0}^{\infty}{\frac{1}{1+\tilde{t}x}x^{N_1}e^{-x}}dx}{\left(N_1-1\right)!}\\
		\frac{\partial^2 F_1(\tilde{t},N_1)}{\partial \tilde{t}^2}&= \frac{\int_{0}^{\infty}{\frac{-1}{\left(1+\tilde{t}x\right)^2}x^{N_1+1}e^{-x}}dx}{\left(N_1-1\right)!}.
		\end{align}
	\end{subequations}
\end{Lemma}
\begin{proof}
	See Appendix \ref{Appendix_Int_log}.
\end{proof}

\begin{Lemma}
	For given \begin{small}$\mathbf{G}$\end{small} and \begin{small}$\mathbf{h}$\end{small}, we define
	\begin{equation}\small
	\begin{aligned}
	&\phi(z)\triangleq&\underset{\boldsymbol{\omega}}{\mathop{\min}}\quad&{\boldsymbol\omega}^H \mathbf{G}\mathbf{G}^H\boldsymbol{\omega}\\
	&&\mathop{s.t.}\quad&\boldsymbol{\omega}^H \mathbf{h} \mathbf{h}^H\boldsymbol{\omega}=z\Vert\mathbf{h}\Vert^2, \quad\Vert\boldsymbol{\omega}\Vert = 1.
	\end{aligned}
	\end{equation}
	Then, \begin{small}$\phi(z)$\end{small} is a smooth function over \begin{small}$[0,1]$\end{small} and \begin{small}$\phi(z)$\end{small} is a convex function. Please see Lemma 6 and Lemma 7 in \cite{statistical_channel} for details.
\end{Lemma}

\begin{Proposition}
	\label{rankone}
	The optimal covariance matrix of \begin{small}$\boldsymbol{\Sigma}_s$\end{small} for \begin{small}$C_1(\mathbf{\Sigma}_s, \mathbf{\Theta})$\end{small} is of rank one.
\end{Proposition}
\begin{proof}
	See Appendix \ref{Appendix_rankone}.
\end{proof}

\subsection{Alternating Optimization for $C_1(\mathbf{\Sigma}_s, \mathbf{\Theta})$}
Since \begin{small}$\mathbf{\Sigma}_s$\end{small} has rank one, let us write \begin{small}$\mathbf{\Sigma}_s = \boldsymbol{\omega}\boldsymbol{\omega}^H$\end{small} with \begin{small}$\Vert\boldsymbol{\omega}\Vert^2=1$\end{small}. Then, the ergodic secrecy rate is reduced to
\begin{equation}\small
\label{AO_proof}
\begin{aligned}
C_1(\mathbf{\Sigma}_s, \mathbf{\Theta}) &= \log\left(1+\rho_r\mathbf{h}_r^H \mathbf{\Theta} \mathbf{G}^H\boldsymbol{\omega}\boldsymbol{\omega}^H \mathbf{G} \mathbf{\Theta}^H \mathbf{h}_r\right) \\
&\quad-
\Expectation_{\mathbf{H}_e} \left \{\log\left[\det \left( \mathbf{I}+\rho_e\mathbf{H}_e^H \mathbf{G}^H\boldsymbol{\omega}\boldsymbol{\omega}^H \mathbf{G} \mathbf{H}_e\right)\right]\right\}\\
&= \log\left(1+\rho_r\mathbf{h}_r^H \mathbf{\Theta} \mathbf{G}^H\boldsymbol{\omega}\boldsymbol{\omega}^H \mathbf{G} \mathbf{\Theta}^H \mathbf{h}_r\right) \\
&\quad-F_1\left(\rho_e\boldsymbol{\omega}^H \mathbf{G}\mathbf{G}^H\boldsymbol{\omega}, N_e\right).
\end{aligned}
\end{equation}
\begin{Proposition}
	For any given beamforming \begin{small}$\boldsymbol{\omega}$\end{small}, the optimal phase shifts are given by
	\begin{equation}\small
	\label{w2v}
	\mathbf{v}=\exp\left(j\arg\left(\diag\left(\mathbf{h}_r^H\right) \mathbf{G}^H\boldsymbol{\omega} \right)+ \mathbf{a}\right)
	\end{equation} where \begin{small}$\mathbf{a}$\end{small} is a vector with all same element.
\end{Proposition}
\begin{proof}
	Invoking (\ref{AO_proof}) and Cauchy-Schwarz inequality, we have
\begin{equation}\small
\begin{aligned}
&\log\left(1+\rho_r\mathbf{h}_r^H \mathbf{\Theta} \mathbf{G}^H\boldsymbol{\omega}\boldsymbol{\omega}^H \mathbf{G} \mathbf{\Theta}^H \mathbf{h}_r\right) -F_1\left(\rho_e\boldsymbol{\omega}^H \mathbf{G}\mathbf{G}^H\boldsymbol{\omega}, N_e\right)\\
&=\log\left(1+\rho_r\mathbf{v}^H \diag\left(\mathbf{h}_r^H\right) \mathbf{G}^H\boldsymbol{\omega}\boldsymbol{\omega}^H \mathbf{G} \diag\left(\mathbf{h}_r\right)\mathbf{v}\right) \\
&\quad - F_1\left(\rho_e\boldsymbol{\omega}^H \mathbf{G}\mathbf{G}^H\boldsymbol{\omega}, N_e\right)\\
&\leq\log\left(1+\rho_r\Vert \diag\left(\mathbf{h}_r^H\right) \mathbf{G}^H\boldsymbol{\omega}\Vert_1^2\right) - F_1\left(\rho_e\boldsymbol{\omega}^H \mathbf{G}\mathbf{G}^H\boldsymbol{\omega}, N_e\right).
\end{aligned}
\end{equation}
	In the last inequality, the equality is achieved at
	\begin{small}$\mathbf{v} = \exp\left(j\arg\left(\diag\left(\mathbf{h}_r^H\right) \mathbf{G}^H\boldsymbol{\omega} \right)+ \mathbf{a}\right)$\end{small}.
\end{proof}
For the propose of accelerating the optimization of \begin{small}$\boldsymbol{\omega}$\end{small}, an
auxiliary variable $z$ is introduced. Let \begin{small}$\mathbf{h}(\mathbf{v}) = \mathbf{G} \diag\left(\mathbf{h}_r\right)\mathbf{v}$\end{small} and \begin{small}$\boldsymbol{\omega}^H \mathbf{h}(\mathbf{v}) \mathbf{h}^H(\mathbf{v}) \boldsymbol{\omega} =z\Vert\mathbf{h}(\mathbf{v})\Vert^2$\end{small}, where \begin{small}$z\in[0,1]$\end{small}. For any given \begin{small}$\mathbf{v}$\end{small} and $z$, the aim of maximizing \begin{small}$C_1(\mathbf{\Sigma}_s, \mathbf{\Theta})$\end{small} motivates us to write
\begin{equation}\small
\label{vz2w}
\begin{aligned}
(P6)\,\,\phi(\mathbf{v}, z)\triangleq&\mathop{\min}_{\boldsymbol{\omega}}&& \boldsymbol{\omega}^H \mathbf{G}\mathbf{G}^H\boldsymbol{\omega}\\
&\mathop{\text{s.t.}}&&\boldsymbol{\omega}^H \mathbf{h}(\mathbf{v}) \mathbf{h}^H(\mathbf{v}) \boldsymbol{\omega} =z\Vert\mathbf{h}(\mathbf{v})\Vert^2\\
&&&\Vert\boldsymbol{\omega}\Vert^2 = 1.
\end{aligned}
\end{equation}
The above optimization implies that we can reduce the eavesdropper's rate while guaranteeing the legitimate receiver's performance, thereby improving secrecy rate. Then, the original beamforming and phase shift problem is reduced to the optimization with respect to \begin{small}$\mathbf{v}$\end{small} and $z$:
\begin{equation}\small
\label{updatez}
\begin{aligned}
(P7)\,\,\widetilde{C}_1(\mathbf{v}, z) \triangleq \mathop{\max}_{\mathbf{v}, z}\,\,&\log\left(1+\rho_r z\Vert\mathbf{h}(\mathbf{v})\Vert^2\right) \\
&\quad- F_1\left(\rho_e\phi(\mathbf{v}, z),N_e\right).
\end{aligned}
\end{equation}
Since \begin{small}$\phi(\mathbf{v}, z)$\end{small} is a smooth function, any converge method (e.g. Newton-type method) can be used to solve it. Note the first- and second- order derivatives can be determined numerically \cite[\S 25.3]{Handbook first second order}. For any given \begin{small}$\mathbf{v}$\end{small}, a Newton-type method is taken to update $z$, thereby optimizing \begin{small}$\boldsymbol{\omega}$\end{small} according to (\ref{vz2w}). Based on (\ref{w2v}), (\ref{vz2w}) and (\ref{updatez}), the overall algorithm proposed in this section is summarized in Algorithm 4.
\begin{algorithm}
	\caption{Alternating Optimization for $C_1(\mathbf{\Sigma}_s, \mathbf{\Theta})$}
	{\bf Input: } \begin{small}$\mathbf{v}^0, t=1$\end{small}
	\begin{algorithmic}[1]
		\Repeat
		\State Update \begin{small}$z^t$\end{small} and \begin{small}$\boldsymbol{\omega}^{t}$\end{small} for given \begin{small}$\mathbf{v}^{t-1}$\end{small} according to (\ref{vz2w}) and (\ref{updatez}) with Newton-type method;
		\State Update \begin{small}$\mathbf{v}^t$\end{small} for given \begin{small}$\boldsymbol{\omega}^t$\end{small} according to (\ref{w2v});
		\State \begin{small}$t \leftarrow t+1$\end{small};
		\Until{Convergence}
	\end{algorithmic}
	{\bf Output: } \begin{small}$\boldsymbol{\omega}=\boldsymbol{\omega}^t, \mathbf{v}=\mathbf{v}^t$\end{small}
\end{algorithm}

\subsection{Analysis on $C_4(\mathbf{\Sigma}_s, \mathbf{\Sigma}_z, \mathbf{\Theta})$}
In this subsection, we study the system with both i.i.d. Gaussian fading channels at receiver and eavesdropper and AN-aided transmission signal. Now, we consider a special scenario where the eavesdropper is equipped with a single antenna, i.e., \begin{small}$N_e = 1$\end{small}.
Therefore, the secrecy rate of  \begin{small}$C_4(\mathbf{\Sigma}_s, \mathbf{\Sigma}_z, \mathbf{\Theta})$\end{small} reduces to:
\begin{equation}\small
\label{C4Special1}
\begin{aligned}
C_4(\mathbf{\Sigma}_s, \mathbf{\Sigma}_z, \mathbf{\Theta})&= \Expectation_{\mathbf{h}_r} \left[\log\left(1+\frac{\rho_r\mathbf{h}_r^H \mathbf{G}^H\mathbf{\Sigma}_s \mathbf{G} \mathbf{h}_r}{1+\rho_r\mathbf{h}_r^H \mathbf{G}^H\mathbf{\Sigma}_z \mathbf{G} \mathbf{h}_r}\right)\right] \\
&\quad- \Expectation_{\mathbf{h}_e} \left[\log\left(1+\frac{\rho_e\mathbf{h}_e^H \mathbf{G}^H\mathbf{\Sigma}_s \mathbf{G} \mathbf{h}_e}{1+\rho_e\mathbf{h}_e^H \mathbf{G}^H\mathbf{\Sigma}_z \mathbf{G} \mathbf{h}_e}\right)\right]\\
&= \Expectation_{\mathbf{h}_e} \left[\log\left(1+\frac{\frac{\rho_r}{\rho_e}\rho_e\mathbf{h}_e^H \mathbf{G}^H\mathbf{\Sigma}_s \mathbf{G} \mathbf{h}_e}{1+\frac{\rho_r}{\rho_e}\rho_e\mathbf{h}_e^H \mathbf{G}^H\mathbf{\Sigma}_z \mathbf{G} \mathbf{h}_e}\right)\right] \\
&\quad- \Expectation_{\mathbf{h}_e} \left[\log\left(1+\frac{\rho_e\mathbf{h}_e^H \mathbf{G}^H\mathbf{\Sigma}_s \mathbf{G} \mathbf{h}_e}{1+\rho_e\mathbf{h}_e^H \mathbf{G}^H\mathbf{\Sigma}_z \mathbf{G} \mathbf{h}_e}\right)\right]
\end{aligned}
\end{equation}
where \begin{small}$\mathbf{h}_e$\end{small} denotes \begin{small}$\mathbf{H}_e$\end{small} with \begin{small}$N_e = 1$\end{small}. The last equality follows from the fact that altering the representation of a random variable does not affect the expectation.

For any given \begin{small}$\mathbf{\Sigma}_s$\end{small}, we define a constant \begin{small}$a\triangleq\rho_e\mathbf{h}_e^H \mathbf{G}^H\mathbf{\Sigma}_s \mathbf{G} \mathbf{h}_e$\end{small} and \begin{small}$\hat{t}\left(\mathbf{\Sigma}_z\right)\triangleq\rho_e\mathbf{h}_e^H \mathbf{G}^H\mathbf{\Sigma}_z \mathbf{G} \mathbf{h}_e$\end{small}, and \begin{small}$b=\rho_e/\rho_r\leq 1$\end{small}. The problem of maximizing \begin{small}$C_4(\mathbf{\Sigma}_s, \mathbf{\Sigma}_z, \mathbf{\Theta})$\end{small} is equivalent to:
\begin{equation}\small
\label{C4Special2}
\begin{aligned}
\underset{\mathbf{\Sigma}_z}{\max} & &\Expectation_{\mathbf{h}_e} \bigg[\log\left(1+\frac{\mathbf{h}_e^H \mathbf{G}^H\mathbf{\Sigma}_s \mathbf{G} \mathbf{h}_e}{b+\mathbf{h}_e^H \mathbf{G}^H\mathbf{\Sigma}_z \mathbf{G} \mathbf{h}_e}\right)\\
&&-\log\left(1+\frac{\mathbf{h}_e^H \mathbf{G}^H\mathbf{\Sigma}_s \mathbf{G} \mathbf{h}_e}{1+\mathbf{h}_e^H \mathbf{G}^H\mathbf{\Sigma}_z \mathbf{G} \mathbf{h}_e}\right)\bigg].\\
\end{aligned}
\end{equation}
By randomly choosing a possible value of \begin{small}$\mathbf h_e$\end{small}, we have
\begin{equation}\small
\begin{aligned}
\text{g}\left(\hat{t}\left(\mathbf{\Sigma}_z\right)\right)&\triangleq\log\left(1+\frac{a}{b+\hat{t}\left(\mathbf{\Sigma}_z\right)}\right) - \log\left(1+\frac{a}{1+\hat{t}\left(\mathbf{\Sigma}_z\right)}\right)\\ \frac{\mathrm{d}\text{g}\left(\hat{t}\left(\mathbf{\Sigma}_z\right)\right)}{\mathrm{d}\hat{t}\left(\mathbf{\Sigma}_z\right)}&\leq 0.
\end{aligned}
\end{equation}
The monotonically decreasing function with respect to \begin{small}$t\left(\mathbf{\Sigma}_z\right)$\end{small} implies that the optimal \begin{small}$\mathbf{\Sigma}_z$\end{small} should be \begin{small}$\mathbf{0}$\end{small}. Intuitively, neither legitimate receiver nor eavesdropper can directly distinguish noise from the overall received signal so that AN is not helpful to the secrecy rate. Hence, the best policy is that when the legitimate channel is better than the eavesdropper channel, all power should be allocated to the information-carrying signal, and if the channel condition is reversed, communication should break up.

\section{Results and Discussions}
\subsection{Simulation Setup}
In this section, we numerically evaluate the performance of the proposed algorithms. In Fig.\ref{Simulation_setup}, we take the Cartesian coordinate system to describe positions of all components in the LIS-enhanced system model. In practical systems, one of the motivations for deploying LIS in secure wireless systems is to enhance favorable communication links by suppressing the undesired eavesdropper. AP and LIS are usually deployed in advance and we assume AP and LIS are located at ($0, 0, 15$)m and ($0, 50, 15$)m, respectively. The channel matrix \begin{small}$\mathbf{G}$\end{small} between AP and LIS channel is modeled as follows
\begin{equation}\small
\mathbf{G}=\sqrt{\frac{K}{1+K}}\mathbf G^{\text{LOS}}+\sqrt{\frac{1}{1+K}}\mathbf G^{\text{NLOS}}
\end{equation}
where the small-scale fading is assumed to be Rician fading. \begin{small}$\mathbf G^{\text{LOS}}$\end{small} and \begin{small}$\mathbf G^{\text{NLOS}}$\end{small} represent the LoS and NLoS components, respectively. The distance-dependent path loss model is given by
\begin{equation}\small
L(d)=C_0\left(\frac{d}{D_0}\right)^{-\zeta}
\end{equation}
where \begin{small}$C_0=-30$\end{small}dB is the path loss at the reference distance \begin{small}$D_0=1$\end{small} meter, $d$ denotes the individual link distance, and \begin{small}$\zeta$\end{small} denotes the path loss exponent. Some parameters are given in Table \ref{Simulation Parameters}.

\begin{figure}
	\centering
	\includegraphics[scale=0.3]{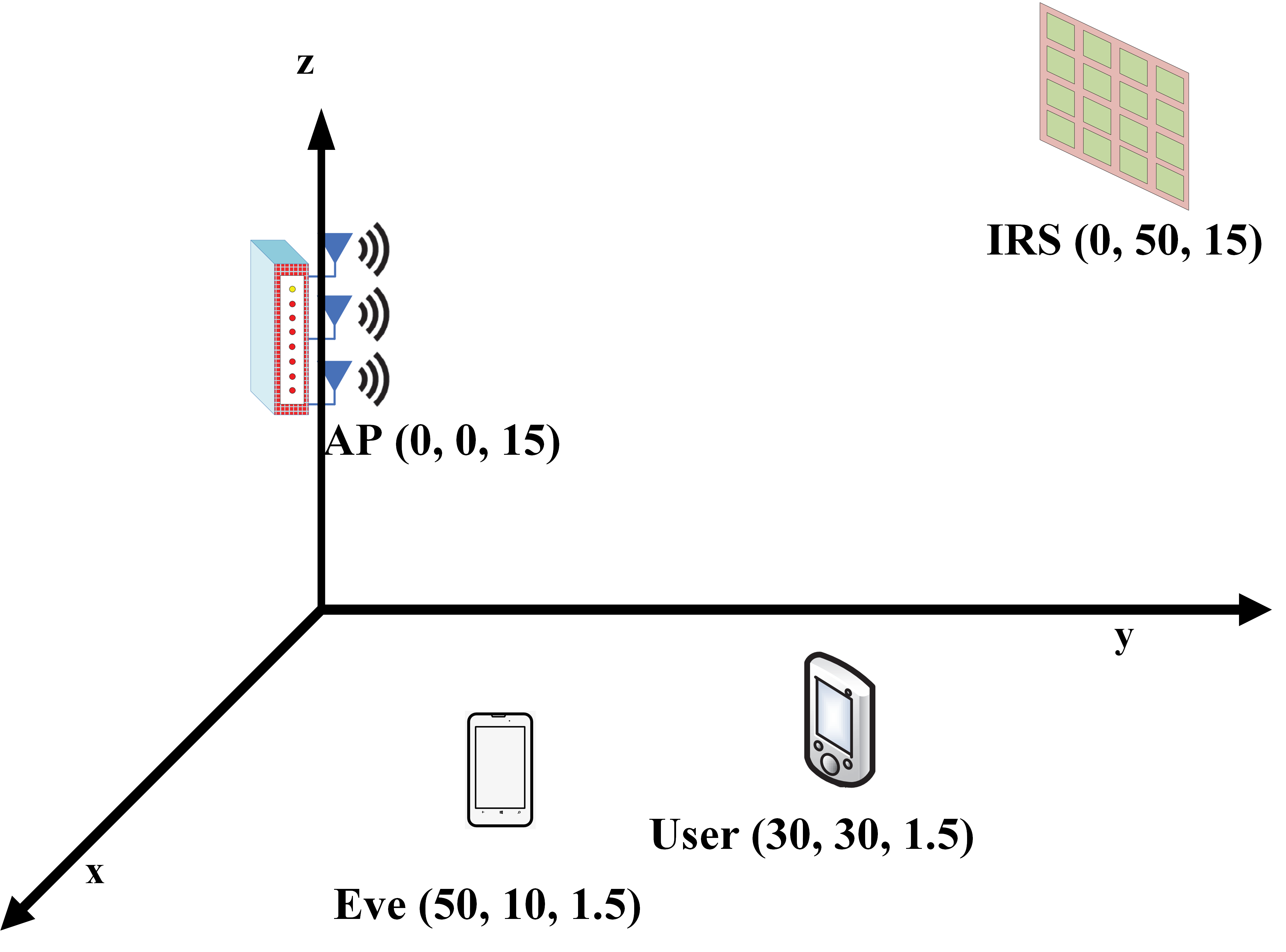}
	\caption{Simulation Setup}
	\label{Simulation_setup}
\end{figure}
\begin{table}
	\scriptsize
	\centering
	\caption{ Simulation Parameters}
	\begin{tabular}{ |c|c| }
		\hline  
		Parameters & Values \\
		\cline{1-2}
		Noise variance, \begin{small}$\sigma^2$\end{small} & $-80$dBm \\
		\cline{1-2}
		The number of antennas at AP, \begin{small}$N_t$\end{small} & $16$\\
		\cline{1-2}
		The number of reflecting elements at LIS, \begin{small}${{N_I}}$\end{small} & $32$\\
		\cline{1-2}
		The number of antennas at Eavesdropper, \begin{small}${{N_e}}$\end{small} & $10$\\  
		\cline{1-2}
		Pass loss in AP-LIS link, \begin{small}$\zeta_{AI}$\end{small} & $2$ \\
		\cline{1-2}
		Pass loss in LIS-Receiver link, \begin{small}$\zeta_{IR}$\end{small} & $2.8$ \\
		\cline{1-2}
		Pass loss in LIS-Eavesdropper link, \begin{small}$\zeta_{IE}$\end{small} & $3$ \\
		\cline{1-2}
		Rician factor, $K$ & $10$\\
		\cline{1-2}
		\hline
	\end{tabular}
	\label{Simulation Parameters}
\end{table}

\subsection{SAA-based vs. Hybrid SPG-CP vs. Alternating Optimization for $C_1(\mathbf{\Sigma}_s, \mathbf{\Theta})$}
In this subsection, we compare the performance and the convergence rates of the SAA-based algorithm, the hybrid SPG-CP algorithm and the proposed alternating optimization algorithm. Fig. \ref{B}(a) depicts that the achievable secrecy rate versus the transmit power which ranges from 10dBm to 30dBm. The system rates increase monotonically with more transmit power budget. It results from the improved SNR of the whole system by providing additional transmit power. Also, we observe that the average secrecy rate achieved by three algorithms is the same. In Fig. \ref{B}(b), we compare the convergence of three algorithms when the total transmit power is 10dBm.  Fig. \ref{B}(b) implies that alternating optimization algorithm makes full use of expectation information and achieves a faster convergence rate than the other two algorithms.
\begin{figure}[!t]
	\centering
	\includegraphics[width=0.48\linewidth]{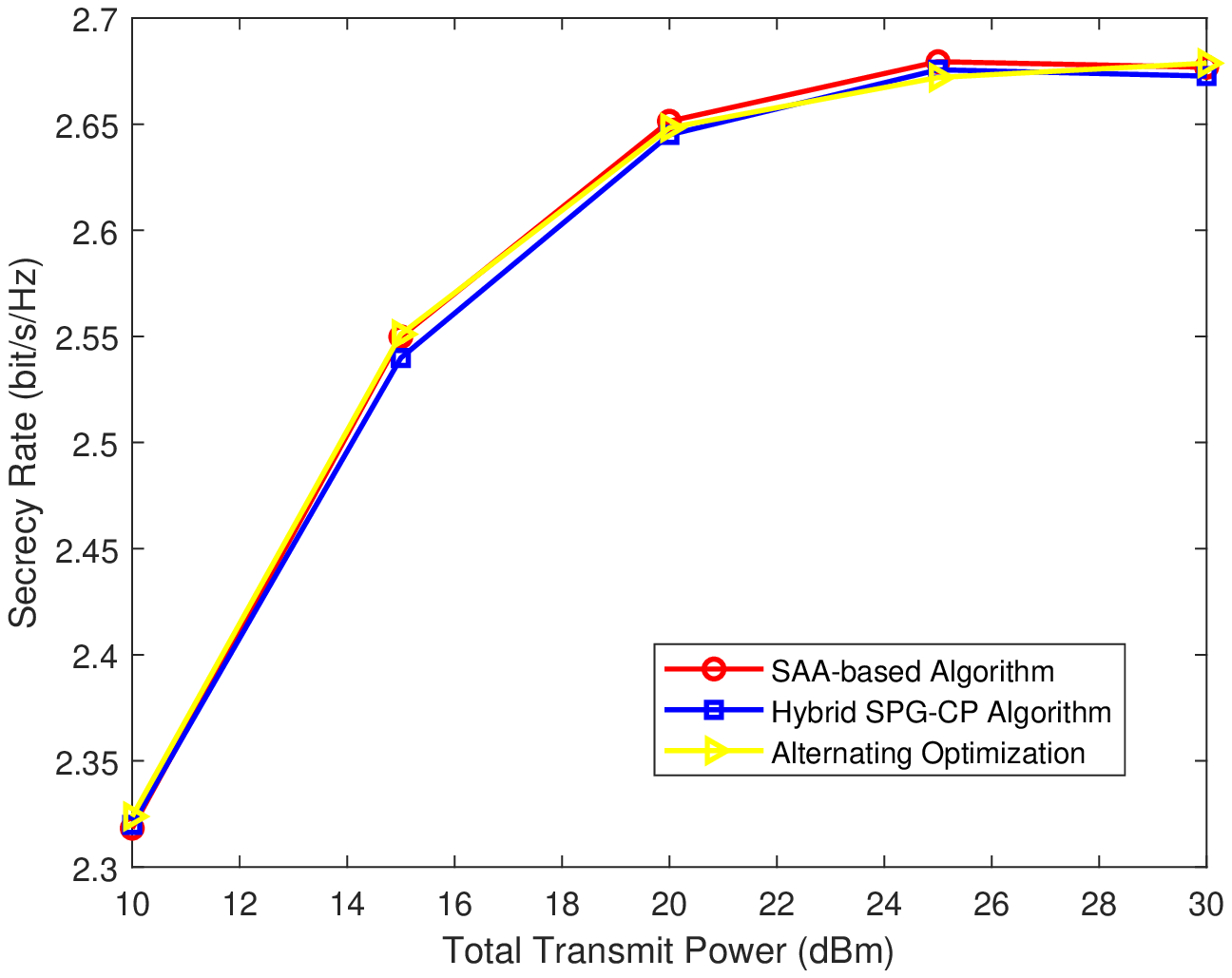} \label{Ba}
	\hfil
	\includegraphics[width=0.48\linewidth]{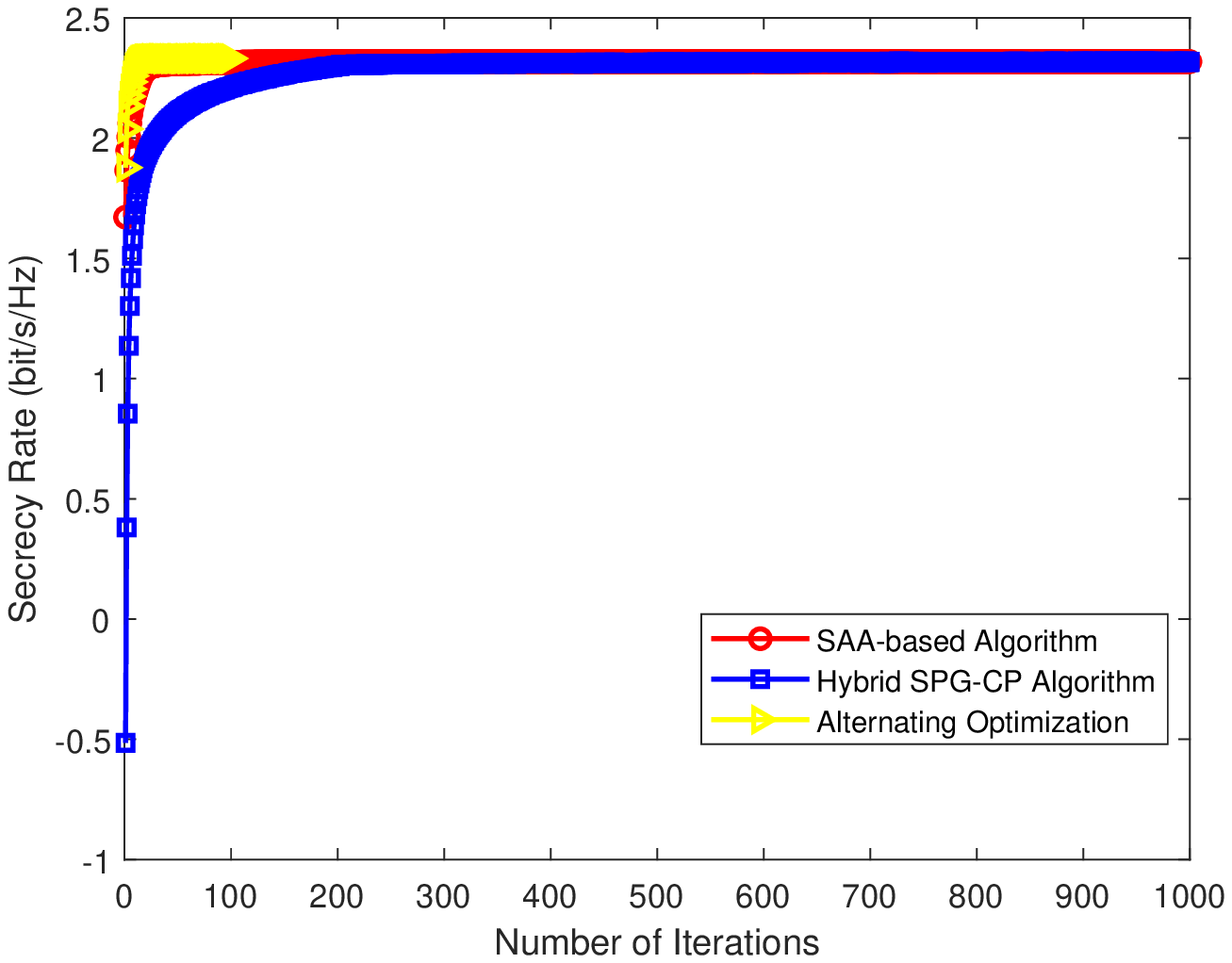} \label{Bb}
	\caption{(a) Secrecy rate (bits/s/Hz) versus the total transmit power of the AP (dBm). (b) Convergence of three algorithms.}
	\label{B}
\end{figure}
\subsection{SAA-based vs. Hybrid SPG-CP for $C_1(\mathbf{\Sigma}_s, \mathbf{\Theta})$ and $C_3(\mathbf{\Sigma}_s, \mathbf{\Sigma}_z, \mathbf{\Theta})$}
For fairness, Fig. \ref{C} compares the non-AN system \begin{small}$C_1(\mathbf{\Sigma}_s, \mathbf{\Theta})$\end{small} and AN-aided system \begin{small}$C_3(\mathbf{\Sigma}_s, \mathbf{\Sigma}_z, \mathbf{\Theta})$\end{small} under the same condition. Fig. \ref{C}(a) shows that the achievable secrecy rate \begin{small}$C_1(\mathbf{\Sigma}_s, \mathbf{\Theta})$\end{small} exceeds the $x$-axis baseline, indicating that a performance gain can be achieved by introducing LIS into the scenario where the direct link is severely blocked. Meanwhile, the achievable secrecy rate \begin{small}$C_3(\mathbf{\Sigma}_s, \mathbf{\Sigma}_z, \mathbf{\Theta})$\end{small} exceeds the \begin{small}$C_1(\mathbf{\Sigma}_s, \mathbf{\Theta})$\end{small} baseline, indicating that the introduced AN-aided structure achieves even greater performance gain due to AN suppressing the channel gain of the eavesdropper. Therefore, the AN-aided structure significantly outperforms the non-AN structure. Fig. \ref{C}(b) sketches the number of iterations versus the secrecy rate by considering three cases with configurations given by: case (1) \begin{small}$\mathbf\Sigma_s=0.5\mathbf I_{N_t}, \mathbf\Sigma_z=0.5\mathbf I_{N_t}, \mathbf \theta_i= 0,\forall i$\end{small}; case (2) \begin{small}$\mathbf\Sigma_s=\mathbf I_{N_t}, \mathbf\Sigma_z=\mathbf 0, \mathbf \theta_i= \pi,\forall i$\end{small}; case (3) \begin{small}$\mathbf\Sigma_s, \mathbf\Sigma_z, \boldsymbol \theta$\end{small} are random matrix. To verify the sensitivity of algorithms with respect to the selection of the initial point, we set the total transmit power \begin{small}$P=15$\end{small}dBm for SAA-based algorithm in the three cases when \begin{small}$P=25$\end{small}dBm for the hybrid SPG-CP algorithm. It can be seen that they always converge to the similar secrecy rates, the differences between them are almost negligible in all the considered cases. The simulation results illustrate that the proposed two algorithms are accurate to optimize the objective problems.
\begin{figure}[!t]
	\centering
	\subfloat[]{\includegraphics[width=0.48\linewidth]{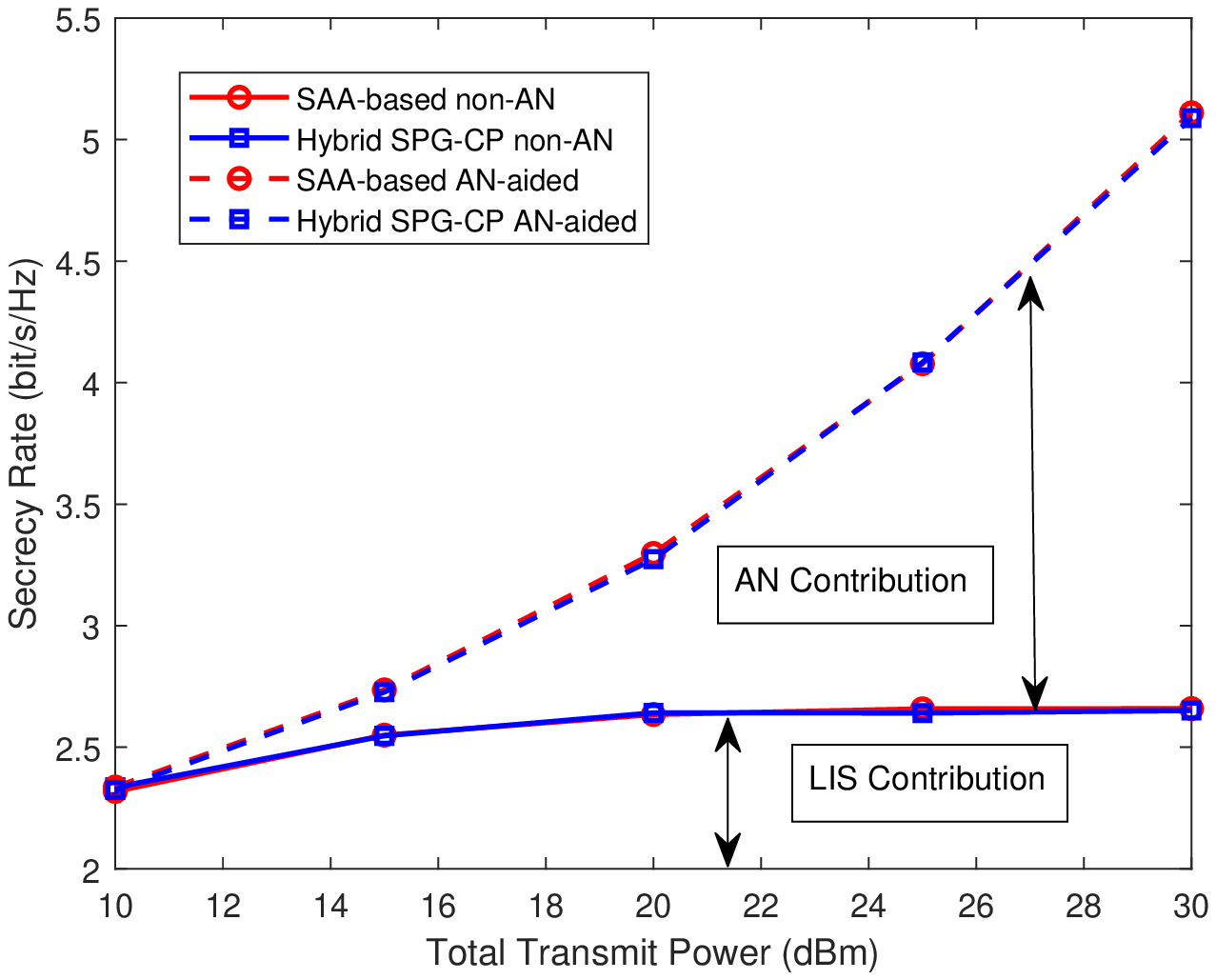} \label{Da}}
	\hfil
	\subfloat[]{\includegraphics[width=0.48\linewidth]{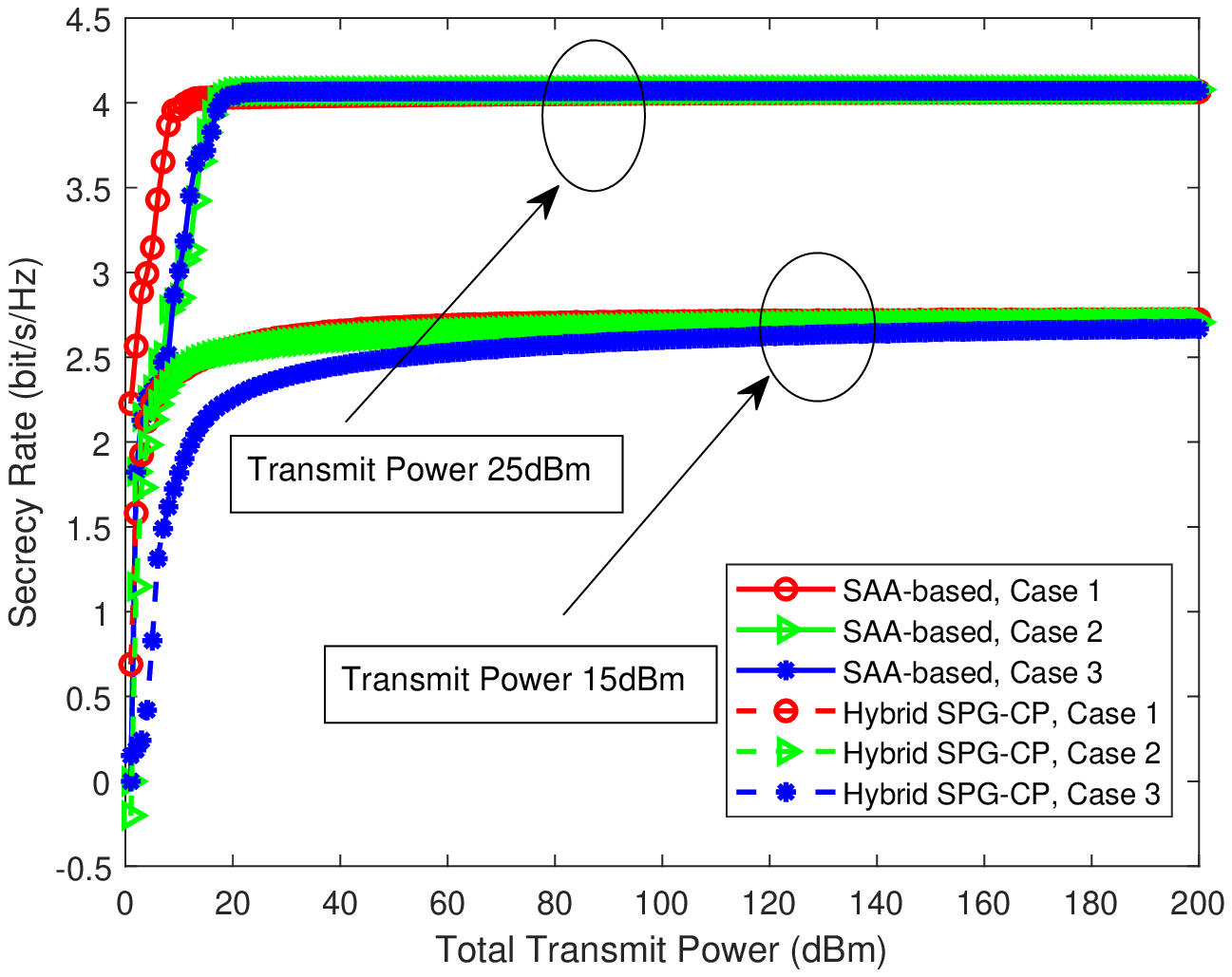} \label{Db}}
	\caption{(a) Secrecy rate (bits/s/Hz) versus the total transmit power of the AP (dBm). (b) Convergence of the SAA-based algorithm and the hybrid SPG-CP algorithm.}
	\label{C}
\end{figure}

\subsection{Hybrid SPG-CP Algorithm for $C_2(\mathbf{\Sigma}_s, \mathbf{\Theta})$ and $C_4(\mathbf{\Sigma}_s, \mathbf{\Sigma}_z, \mathbf{\Theta})$ with $N_e = 1$}
In this subsection, for the purpose of demonstrating that AN structure has no contribution to the system with both i.i.d. Gaussian channel models and \begin{small}$N_e = 1$\end{small}, we compare the secrecy rates \begin{small}$C_2(\mathbf{\Sigma}_s, \mathbf{\Theta})$\end{small} and \begin{small}$C_4(\mathbf{\Sigma}_s, \mathbf{\Sigma}_z, \mathbf{\Theta})$\end{small}. Fig. \ref{E} sketches the transmit power versus the secrecy rate by considering two cases with configurations given by: case (1) \begin{small}$\zeta_{IR}=2.8$\end{small}; case (2) \begin{small}$\zeta_{IR}=2.2$\end{small}. Simulation results show that the achievable secrecy rate is the same in both cases, which implies the power allocated to AN is 0. Besides, we also compare secrecy rates by setting different path losses. As can be observed from Fig. \ref{E}, increasing transmit power results in an improved secrecy rate. Therefore, we should allocate all power to beamforming when both i.i.d. Gaussian channel models are considered.
\begin{figure}[!t]
	\centering
	\includegraphics[width=0.7\linewidth]{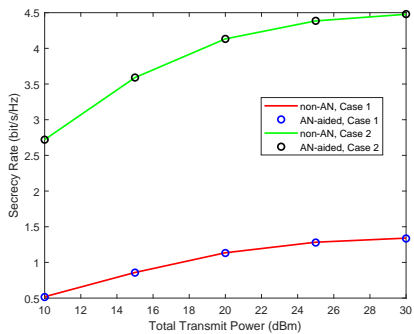}
	\caption{Secrecy rate (bits/s/Hz) versus the total transmit power of the AP (dBm) with AN-aided/non-AN structure. Set $\zeta_{IR}=2.8$ in case 1 and $\zeta_{IR}=2.2$ in case 2.}
	\label{E}
\end{figure}

\subsection{Secrecy Rate vs. Number of LIS elements}
In Fig. \ref{F}, we study the impact of the number of LIS elements under two cases: (1) transmit power \begin{small}$P=15$\end{small}dBm; (2) transmit power \begin{small}$P=25$\end{small}dBm. We assume that the number ranges from \begin{small}$[8,40]$\end{small}. As clearly shown, both proposed algorithms yield very similar performance curves and the differences between them are almost negligible. Furthermore, it is observed that as the number of LIS elements increases, the secrecy rate increases accordingly because more LIS elements can reflect more energy. Therefore, the LIS-based system contributes to improving performance significantly.
\begin{figure}[!t]
	\centering
	\includegraphics[width=0.7\linewidth]{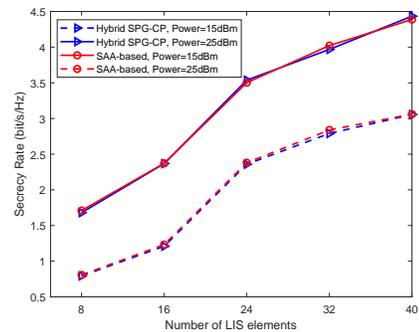}
	\caption{Secrecy rate (bits/s/Hz) versus Number of LIS elements. Set transmit power $P=15$dBm and $P=25$dBm.}
	\label{F}
\end{figure}

\section{Conclusions}
In this paper, we have considered the problem of improving the physical layer security of wireless communication networks by deploying LIS. Two efficient algorithms, i.e., the SAA-based algorithm and the hybrid SPG-CP algorithm, were proposed for joint optimization of the AN-aided beamforming at the transmitter and phase shifts at the LIS. The SAA-based algorithm has high complexity at each iteration due to a log-like function existing in the objective function, while the hybrid SPG-CP algorithm has closed-form solution at each iteration, but needs more iterations. Also, analyses were provided on \begin{small}$C_1(\mathbf{\Sigma}_s, \mathbf{\Theta})$\end{small} and \begin{small}$C_4(\mathbf{\Sigma}_s, \mathbf{\Sigma}_z, \mathbf{\Theta})$\end{small}. An alternating optimization algorithm was proposed to solve \begin{small}$C_1(\mathbf{\Sigma}_s, \mathbf{\Theta})$\end{small} more efficiently. For \begin{small}$C_4(\mathbf{\Sigma}_s, \mathbf{\Sigma}_z, \mathbf{\Theta})$\end{small}, our result showed that AN is not necessary for both i.i.d. Gaussian fading channel in LIS-receiver and LIS-eavesdropper links. Simulation results showed the larger number of LIS elements we use, the better performance the system achieves. Furthermore, we confirmed the huge potential of LIS in improving security and energy efficiency in future communication systems.

\appendices
\section{Proof that $\nabla f_e^i\left(\mathbf{X}\right)$ is Lipschitz continuous}
\label{Estimation Lipschitz constant}
To begin with, for notational simplicity, we rewrite \begin{small}$\nabla f_e^i\left(\mathbf{X}\right)$\end{small} in the following
general form:
\begin{equation}\small
\nabla f_e^i\left(\mathbf{X}\right) =\left( \rho_e\mathbf{H}^i\left(\mathbf{I}+\rho_e\mathbf{H}^{iH}\mathbf{X}\mathbf{H}^i \right)^{-1}\mathbf{H}^{iH}\right)^T
\end{equation}
where \begin{small}$\mathbf{H}^i= \mathbf{G}\mathbf{H}_e^i$\end{small}. Then, for the Lipschitz continuous of the gradient, we have
\begin{equation}\small
\label{Estimation1}
\begin{aligned}
&\Vert\nabla f_e^i\left(\mathbf{X}\right)-\nabla f_e^i\left(\mathbf{Y}\right)\Vert \\
&=\bigg\Vert \rho_e\mathbf{H}^i\left(\left(\mathbf{I}+\rho_e\mathbf{H}^{iH}\mathbf{X}\mathbf{H}^i \right)^{-1}-\left(\mathbf{I}+\rho_e\mathbf{H}^{iH}\mathbf{Y}\mathbf{H}^i \right)^{-1}\right)\mathbf{H}^{iH}\bigg\Vert\\
&\overset{(a)}{=}\bigg\Vert \rho_e^2\mathbf{H}^i\left(\mathbf{I}+\rho_e\mathbf{H}^{iH}\mathbf{X}\mathbf{H}^i \right)^{-1}\mathbf{H}^{iH}(\mathbf{Y}-\mathbf{X})\\
&\quad\quad\quad\mathbf{H}^i\left(\mathbf{I}+\rho_e\mathbf{H}^{iH}\mathbf{Y}\mathbf{H}^i \right)^{-1}\mathbf{H}^{iH}\bigg\Vert\\
&=\Vert\nabla f_e^i\left(\mathbf{X}\right)(\mathbf{Y}-\mathbf{X})^T \nabla f_e^i\left(\mathbf{Y}\right)\Vert\\
&\overset{(b)}{\leq}\Vert\nabla f_e^i\left(\mathbf{X}\right)\Vert\Vert\mathbf{X}-\mathbf{Y}\Vert\Vert\nabla f_e^i\left(\mathbf{Y}\right)\Vert\\
&\leq \left(\mathop{\max}_\mathbf{X} \Vert\nabla f_e^i\left(\mathbf{X}\right)\Vert\right)^2\Vert\mathbf{X}-\mathbf{Y}\Vert
\end{aligned}
\end{equation}
where (a) follows from \begin{small}$\mathbf{A}^{-1}-\mathbf{B}^{-1}=\mathbf{A}^{-1}(\mathbf{B}-\mathbf{A})\mathbf{B}^{-1}$\end{small}; (b) is due to \begin{small}$\Vert\mathbf{A}\mathbf{B}\Vert\leq\Vert\mathbf{A}\Vert\Vert\mathbf{B}\Vert$\end{small}. Next, an upper-bound of \begin{small}$\Vert\nabla f_e^i\left(\mathbf{X}\right)\Vert$\end{small} is given as follows:
\begin{equation}\small
\label{Estimation2}
\begin{aligned}
\Vert\nabla f_e^i\left(\mathbf{X}\right)\Vert &\leq \rho_e \Vert \mathbf{H}^i\mathbf{H}^{iH}\Vert\Vert\left(\mathbf{I}+\rho_e\mathbf{H}^{iH}\mathbf{X}\mathbf{H}^i \right)^{-1}\Vert\\
&\overset{(c)}{\leq}\rho_e \Vert \mathbf{H}^i\mathbf{H}^{iH}\Vert\Tr\left[\left(\mathbf{I}+\rho_e\mathbf{H}^{iH}\mathbf{X}\mathbf{H}^i \right)^{-1}\right]\\
&\overset{(d)}{\leq}\rho_e N_e \Vert \mathbf{H}^i\mathbf{H}^{iH}\Vert\\
\end{aligned}
\end{equation}
where (c) follows from \begin{small}$\sqrt{\Tr(\mathbf{A}\mathbf{B})}\leq\frac{1}{2}(\Tr(\mathbf{A})+\Tr(\mathbf{B}))$\end{small} and non-negative definite matrices \begin{small}$\mathbf{A}, \mathbf{B}$\end{small}; (d) is due to the fact that the maximal eigenvalue of  \begin{small}$\left(\mathbf{I}+\rho_e\mathbf{H}^{iH}\mathbf{X}\mathbf{H}^i \right)^{-1}$\end{small} is no more than $1$. For any \begin{small}$L_i\geq \left(\rho_e N_e \Vert \mathbf{H}^i\mathbf{H}^{iH}\Vert\right)^2=\left(\rho_e N_e \Vert \mathbf{G}\mathbf{H}_e^i\mathbf{H}_e^{iH} \mathbf{G}^H\Vert\right)^2$\end{small}, \begin{small}$\Vert\nabla f_e^i\left(\mathbf{X}\right)-\nabla f_e^i\left(\mathbf{Y}\right)\Vert\leq L_i\Vert\mathbf{X}-\mathbf{Y}\Vert$\end{small} holds. Therefore, \begin{small}$\nabla f_e^i\left(\mathbf{X}\right)$\end{small} is Lipschitz continuous.

\section{Proof of Proposition \ref{stateL1}}
\label{Appendix_stateL1}
We first introduce the following proposition. It is noted that the equalities (\ref{assumptions_1})-(\ref{assumptions_4}) are readily satisfied if \begin{small}$\mathbf{H}_{e,t}^i$\end{small} is bounded \cite[Assumption C]{PSCA}. However, each element in \begin{small}$\mathbf{H}_{e,t}^i$\end{small} is i.i.d. complex normal distribution distributed with zero mean and unit variance, which implies that \begin{small}$\mathbf{H}_{e,t}^i$\end{small} is unbounded. The following proposition extends the equalities in \cite{PSCA} to the unbounded \begin{small}$\mathbf{H}_{e,t}^i$\end{small} case.
\begin{Proposition}
	\label{unbiased}
	For any \begin{small}$\left(\mathbf{\Sigma}_s^{t}, \mathbf{\Sigma}_z^{t}, \mathbf{\Theta}^t\right)$\end{small} generated in our algorithm, we have
	\begin{subequations}\small
		\label{assumptions}
		\begin{align}
		&\Expectation\left[\nabla_{\mathbf{\Sigma}_s^{*}}C_3(\mathbf{\Sigma}_s^{t}, \mathbf{\Sigma}_z^{t}, \mathbf{\Theta}^t,\mathbf{H}_{e,t}^i)-\nabla_{\mathbf{\Sigma}_s^{*}}C_{3}\left(\mathbf{\Sigma}_s^{t},\mathbf{\Sigma}_z^{t},\mathbf{\Theta}^t\right)\right]=0
		\label{assumptions_1}\\
		&\Expectation\left[\nabla_{\mathbf{\Sigma}_z^{*}}C_3(\mathbf{\Sigma}_s^{t}, \mathbf{\Sigma}_z^{t}, \mathbf{\Theta}^t,\mathbf{H}_{e,t}^i)-\nabla_{\mathbf{\Sigma}_z^{*}}C_{3}\left(\mathbf{\Sigma}_s^{t},\mathbf{\Sigma}_z^{t},\mathbf{\Theta}^t\right)\right]=0\\
		&\Expectation\left[\Vert\nabla_{\mathbf{\Sigma}_s^{*}}C_3(\mathbf{\Sigma}_s^{t}, \mathbf{\Sigma}_z^{t}, \mathbf{\Theta}^t,\mathbf{H}_{e,t}^i)-\nabla_{\mathbf{\Sigma}_s^{*}}C_{3}\left(\mathbf{\Sigma}_s^{t},\mathbf{\Sigma}_z^{t},\mathbf{\Theta}^t\right)\Vert^2\right]<\sigma^2\\
		&\Expectation\left[\Vert\nabla_{\mathbf{\Sigma}_z^{*}}C_3(\mathbf{\Sigma}_s^{t}, \mathbf{\Sigma}_z^{t}, \mathbf{\Theta}^t,\mathbf{H}_{e,t}^i)-\nabla_{\mathbf{\Sigma}_z^{*}}C_{3}\left(\mathbf{\Sigma}_s^{t},\mathbf{\Sigma}_z^{t},\mathbf{\Theta}^t\right)\Vert^2\right]<\sigma^2	\label{assumptions_4}
		\end{align}
	\end{subequations}
	for some constant \begin{small}$\sigma$\end{small}.
\end{Proposition}
\begin{proof}
	See Appendix \ref{Appendix_unbiased}.
\end{proof}
Proceeding as in (\ref{up1})-(\ref{up}), for any given $\mathbf \Theta$, we have
\begin{eqnarray}\small
&&\|\nabla\widehat{C}_3(\widehat{\mathbf x}_1,\mathbf\Theta)-\nabla\widehat{C}_3(\widehat{\mathbf x}_2,\mathbf\Theta)\|^2\nonumber\\
&&=\Vert\nabla_{\mathbf\Sigma_s^*} f_{r,t}(\mathbf{\Sigma}_s^1+ \mathbf{\Sigma}_z^1)-\nabla_{\mathbf\Sigma_s^*} f_{r,t}(\mathbf{\Sigma}_s^2+ \mathbf{\Sigma}_z^2)\nonumber\\
&&\quad-\Expectation\{\nabla_{\mathbf\Sigma_s^*} f_{e,t}^i(\mathbf{\Sigma}_s^1+ \mathbf{\Sigma}_z^1)-\nabla_{\mathbf\Sigma_s^*} f_{e,t}^i(\mathbf{\Sigma}_s^2+ \mathbf{\Sigma}_z^2)\}\Vert^2\nonumber\\
&&\quad+\Vert\nabla_{\mathbf\Sigma_z^*} f_{r,t}(\mathbf{\Sigma}_s^1+ \mathbf{\Sigma}_z^1)-\nabla_{\mathbf\Sigma_z^*} f_{r,t}(\mathbf{\Sigma}_s^2+ \mathbf{\Sigma}_z^2)\nonumber\\
&&\quad-\nabla_{\mathbf\Sigma_z^*} f_{r,t}(\mathbf{\Sigma}_z^1)+\nabla_{\mathbf\Sigma_z^*} f_{r,t}(\mathbf{\Sigma}_z^2)\nonumber\\
&&\quad-\Expectation\{\nabla_{\mathbf\Sigma_z^*} f_{e,t}^i(\mathbf{\Sigma}_s^1+ \mathbf{\Sigma}_z^1)-\nabla_{\mathbf\Sigma_z^*} f_{e,t}^i(\mathbf{\Sigma}_s^2+ \mathbf{\Sigma}_z^2)\}\nonumber\\
&&\quad+\Expectation\{\nabla_{\mathbf\Sigma_z^*} f_{e,t}^i(\mathbf{\Sigma}_z^1)\}-\Expectation\{\nabla_{\mathbf\Sigma_z^*} f_{e,t}^i(\mathbf{\Sigma}_z^2)\}\Vert^2\nonumber\\
&&=2\Vert\nabla_{\mathbf\Sigma_s^*} f_{r,t}(\mathbf{\Sigma}_s^1+ \mathbf{\Sigma}_z^1)-\nabla_{\mathbf\Sigma_s^*} f_{r,t}(\mathbf{\Sigma}_s^2+ \mathbf{\Sigma}_z^2)\Vert^2\nonumber\\
&&\quad+2\Vert\Expectation\{\nabla_{\mathbf\Sigma_s^*} f_{e,t}^i(\mathbf{\Sigma}_s^1+ \mathbf{\Sigma}_z^1)-\nabla_{\mathbf\Sigma_s^*} f_{e,t}^i(\mathbf{\Sigma}_s^2+ \mathbf{\Sigma}_z^2)\}\Vert^2\nonumber\\
&&\quad+4\Vert\nabla_{\mathbf\Sigma_z^*} f_{r,t}(\mathbf{\Sigma}_s^1+ \mathbf{\Sigma}_z^1)-\nabla_{\mathbf\Sigma_z^*} f_{r,t}(\mathbf{\Sigma}_s^2+ \mathbf{\Sigma}_z^2)\Vert^2\nonumber\\
&&\quad+4\Vert\nabla_{\mathbf\Sigma_z^*} f_{r,t}(\mathbf{\Sigma}_z^1)-\nabla_{\mathbf\Sigma_z^*} f_{r,t}(\mathbf{\Sigma}_z^2)\Vert^2\nonumber\\
&&\quad+4\Vert\Expectation\{\nabla_{\mathbf\Sigma_z^*} f_{e,t}^i(\mathbf{\Sigma}_s^1+ \mathbf{\Sigma}_z^1)-\nabla_{\mathbf\Sigma_z^*} f_{e,t}^i(\mathbf{\Sigma}_s^2+ \mathbf{\Sigma}_z^2)\}\Vert^2\nonumber\\
&&\quad+4\Vert\Expectation\{\nabla_{\mathbf\Sigma_z^*} f_{e,t}^i(\mathbf{\Sigma}_z^1)-\nabla_{\mathbf\Sigma_z^*} f_{e,t}^i(\mathbf{\Sigma}_z^2)\}\Vert^2\nonumber\\
&&\leq(16\rho_r^4G^4N_I^4\Vert h_r\Vert^8+256\rho_e^4G^4N_I^8N_e^{16})\nonumber\\
&&\quad(\Vert\mathbf\Sigma_s^1-\mathbf\Sigma_s^2\Vert^2+\Vert\mathbf\Sigma_z^1-\mathbf\Sigma_z^2\Vert^2)\\
&&=(16\rho_r^4G^4N_I^4\Vert h_r\Vert^8+256\rho_e^4G^4N_I^8N_e^{16})\Vert\widehat{\mathbf x}_1-\widehat{\mathbf x}_2\Vert^2\nonumber
\end{eqnarray}
The last inequality follows the same analyses as (\ref{up1}) and (\ref{up}) where the coefficient in front of the norm has been increased sufficiently. Therefore, there exists a constant $L$ irrelevant to $\mathbf\Theta$ such that
\begin{equation}\small
\|\nabla\widehat{C}_3(\widehat{\mathbf x}_1,\mathbf \Theta)-\nabla\widehat{C}_3(\widehat{\mathbf x}_2,\mathbf \Theta)\|\leq L\Vert\widehat{\mathbf x}_1-\widehat{\mathbf x}_2\Vert\quad \forall \widehat{\mathbf x}_1,\widehat{\mathbf x}_2
\end{equation}
Following the same analysis as the Descent Lemma \cite{MM_Proposed}, we have
\begin{equation}\small
\begin{aligned}
-\widehat{C}_3(\widehat{\mathbf x}_1,\mathbf \Theta)\leq&-\widehat{C}_3(\widehat{\mathbf x}_2,\mathbf \Theta)\\
&-2\langle\nabla\widehat{C}_3(\widehat{\mathbf x}_2,\mathbf \Theta),\widehat{\mathbf x}_1-\widehat{\mathbf x}_2\rangle+L\|\widehat{\mathbf x}_1-\widehat{\mathbf x}_2\|^2
\end{aligned}
\end{equation}
Finally, taking the form of \begin{small}$C_{3}\left(\mathbf{\Sigma}_s,\mathbf{\Sigma}_z,\mathbf \Theta\right)$\end{small}, we have
\begin{equation}\small
\begin{aligned} &-C_{3}\left(\mathbf{\Sigma}_s^1,\mathbf{\Sigma}_z^1,\mathbf \Theta\right)\leq -C_{3}\left(\mathbf{\Sigma}_s^2,\mathbf{\Sigma}_z^2,\mathbf \Theta\right)\\
&\quad-2\langle \nabla_{\mathbf{\Sigma}_{s}^*}C_{3}\left(\mathbf{\Sigma}_s^2,\mathbf{\Sigma}_z^2,\mathbf \Theta\right),\mathbf\Sigma_{s}^1-\mathbf\Sigma_s^2\rangle+L\|\mathbf\Sigma_{s}^1-\mathbf\Sigma_s^2\|^2\\
&\quad-2\langle \nabla_{\mathbf{\Sigma}_{z}^*}C_{3}\left(\mathbf{\Sigma}_s^2,\mathbf{\Sigma}_z^2,\mathbf \Theta\right),\mathbf\Sigma_{z}^1-\mathbf\Sigma_z^2\rangle+L\|\mathbf\Sigma_{z}^1-\mathbf\Sigma_z^2\|^2\\
\end{aligned}
\end{equation}

\section{Proof of Proposition \ref{unbiased}}
\label{Appendix_unbiased}
Let $X$ be the maximum of the real and imaginary parts of all the elements in the \begin{small}$\mathbf{H}_{e,t}^i$\end{small}. The real and imaginary parts of \begin{small}$\mathbf{H}_{e,t}^i$\end{small} are i.i.d. Gaussian with mean $0$ and variance $1/2$. Then
\begin{equation}\small
\label{step1}
\begin{aligned}
F_X(x)=Pr(X\leq x)=\left(\frac{1}{\sqrt{\pi}}\int_{-\infty}^x\exp(-t^2)\,dt\right)^{2{{N_I}}N_e}
\end{aligned}
\end{equation}
and the probability density function (p.d.f.) of $X$ satisfies
\begin{equation}\small
\label{step2}
\begin{aligned}
p_X(x)&=\frac{d F_X(x)}{d x}\\&=2{{N_I}}N_e\left(\frac{1}{\sqrt{\pi}}\int_{-\infty}^x\exp(-t^2)\,dt\right)^{2{{N_I}}N_e-1}\frac{1}{\sqrt{\pi}}\exp(-x^2)\\
&\leq\frac{2{{N_I}}N_e}{\sqrt{\pi}}\exp(-x^2)
\end{aligned}
\end{equation}
In addition, we assume that \begin{small}$X=x$\end{small}, then the upper bound of \begin{small}$\Vert\nabla f_{e,t}^i(\mathbf{\Sigma}_s+ \mathbf{\Sigma}_z)\Vert$\end{small} defined in (\ref{tidu_rteti}) is given by
\begin{equation}\small
\label{step3}
\begin{aligned}
&\Vert\nabla f_{e,t}^i(\mathbf{\Sigma}_s+ \mathbf{\Sigma}_z)\Vert\\
&=\Vert \rho_e\mathbf{G}\mathbf{H}_{e,t}^i\left(\mathbf{I}+\rho_e\mathbf{H}_{e,t}^{iH} \mathbf{G}^H(\mathbf{\Sigma}_s+ \mathbf{\Sigma}_z) \mathbf{G}\mathbf{H}_{e,t}^i \right)^{-1}\mathbf{H}_{e,t}^{iH} \mathbf{G}^H\Vert\\
&\leq\rho_e\Vert\mathbf{G}^H\mathbf{G}\Vert\Vert\mathbf{H}_{e,t}^i\mathbf{H}_{e,t}^{iH}\Vert\\
&\quad\Vert\left(\mathbf{I}+\rho_e\mathbf{H}_{e,t}^{iH} \mathbf{G}^H(\mathbf{\Sigma}_s+ \mathbf{\Sigma}_z) \mathbf{G}\mathbf{H}_{e,t}^i \right)^{-1}\Vert\\
&\leq 2x^2\rho_e G {{N_I}} N_e^2 
\end{aligned}
\end{equation}
where \begin{small}$G\triangleq\Vert\mathbf{G}^H\mathbf{G}\Vert$\end{small}; \begin{small}$\Vert\left(\mathbf{I}+\rho_e\mathbf{H}_{e,t}^{iH} \mathbf{G}^H(\mathbf{\Sigma}_s+ \mathbf{\Sigma}_z) \mathbf{G}\mathbf{H}_{e,t}^i \right)^{-1}\Vert\leq N_e$\end{small}; \begin{small}$\Vert\mathbf{H}_{e,t}^i\mathbf{H}_{e,t}^{iH}\Vert\leq 2x^2{{N_I}}N_e$\end{small}, which is derived from the fact that Frobenius norm reaches the maximum as the absolute squares of each element is maximum.

Then, we have
\begin{equation}\small
\label{step4}
\begin{aligned}
&\Expectation\left[\Vert \nabla f_{e,t}^i(\mathbf{\Sigma}_s+ \mathbf{\Sigma}_z)\Vert\right]\leq \frac{4\rho_e G {{N_I}}^2 N_e^3 }{\sqrt{\pi}}\int_{-\infty}^{+\infty}x^2\exp(-x^2)\,dx\\
&=2\rho_e G {{N_I}}^2 N_e^3\\
&\Expectation\left[\Vert \nabla f_{e,t}^i(\mathbf{\Sigma}_s+ \mathbf{\Sigma}_z)\Vert^2\right]\leq \frac{8\rho_e^2 G^2 {{N_I}}^3 N_e^5 }{\sqrt{\pi}}\int_{-\infty}^{+\infty}x^4\exp(-x^2)\,dx\\
&=6\rho_e^2 G^2 {{N_I}}^3 N_e^5.
\end{aligned}
\end{equation}
Similarly, the formula \begin{small}$\Expectation\left[\Vert \nabla f_{e,t}^i( \mathbf{\Sigma}_z)\Vert\right]\leq 2\rho_e G {{N_I}}^2 N_e^3$\end{small} and \begin{small}$\Expectation\left[\Vert \nabla f_{e,t}^i( \mathbf{\Sigma}_z)\Vert^2\right]\leq6\rho_e^2 G^2 {{N_I}}^3 N_e^5$\end{small} hold. Due to
\begin{equation}\small
\label{step5}
\begin{aligned}
&\Expectation\left[\Vert\nabla_{\mathbf{\Sigma}_z^{*}}C_3(\mathbf{\Sigma}_s, \mathbf{\Sigma}_z, \mathbf{\Theta},\mathbf{H}_{e,t}^i)\Vert^2\right]\\
&=\Expectation[\Vert\nabla f_{r,t}(\mathbf{\Sigma}_s+ \mathbf{\Sigma}_z)-\nabla f_{r,t}(\mathbf{\Sigma}_z)-\nabla f_{e,t}^i(\mathbf{\Sigma}_s+ \mathbf{\Sigma}_z)\\
&\quad+\nabla f_{e,t}^i(\mathbf{\Sigma}_z)\Vert^2]\\
&\leq 4\bigg(\Vert\nabla f_{r,t}(\mathbf{\Sigma}_s+ \mathbf{\Sigma}_z)\Vert^2+\Vert\nabla f_{r,t}(\mathbf{\Sigma}_z)\Vert^2\\
&\quad+\Expectation\left[\Vert\nabla f_{e,t}^i(\mathbf{\Sigma}_s+ \mathbf{\Sigma}_z)\Vert^2\right]+\Expectation\left[\Vert\nabla f_{e,t}^i(\mathbf{\Sigma}_z)\Vert^2\right]\bigg)\\
&= 4\left(\Vert\nabla f_{r,t}(\mathbf{\Sigma}_s+ \mathbf{\Sigma}_z)\Vert^2+\Vert\nabla f_{r,t}(\mathbf{\Sigma}_z)\Vert^2+12\rho_e^2 G^2 {{N_I}}^3 N_e^5\right)
\end{aligned}
\end{equation}
\begin{small}$\Expectation\left[\Vert\nabla_{\mathbf{\Sigma}_z^{*}}C_3(\mathbf{\Sigma}_s, \mathbf{\Sigma}_z, \mathbf{\Theta},\mathbf{H}_{e,t}^i)\Vert\right]$\end{small} is bounded. Similarly, \begin{small}$\Expectation\left[\Vert\nabla_{\mathbf{\Sigma}_s^{*}}C_3(\mathbf{\Sigma}_s, \mathbf{\Sigma}_z, \mathbf{\Theta},\mathbf{H}_{e,t}^i)\Vert\right]$\end{small} is bounded.

For any \begin{small}$\mathbf{\Sigma}_s^{n}$\end{small}, we consider a function \begin{small}$R\left(\mathbf{\Sigma}_{z}, \mathbf{H}_{e,t}^i\right)\triangleq f_{e,t}^i(\mathbf{\Sigma}_{s}^n+ \mathbf{\Sigma}_z)p(\mathbf{H}_{e,t}^i)$\end{small} with the joint p.d.f. of each element in \begin{small}$\mathbf{H}_{e,t}^i$\end{small} being \begin{small}$p(\mathbf{H}_{e,t}^i)$\end{small} and its gradient with respect to \begin{small}$\mathbf{\Sigma}_{z}$\end{small} is
\begin{equation}\small
\begin{aligned}
&\nabla R\left(\mathbf{\Sigma}_{z}, \mathbf{H}_{e,t}^i\right)= \rho_e\mathbf{G}\mathbf{H}_{e,t}^i\bigg[\mathbf{I}+\rho_e\mathbf{H}_{e,t}^{iH} \mathbf{G}^H(\mathbf{\Sigma}_s^n+ \mathbf{\Sigma}_z)
\mathbf{G}\mathbf{H}_{e,t}^i \bigg]^{-1}\\&\quad\mathbf{H}_{e,t}^{iH} \mathbf{G}^Hp(\mathbf{H}_{e,t}^i).
\end{aligned}
\end{equation}
and an upper bound of \begin{small}$\Vert\nabla R\left(\mathbf{\Sigma}_{s}, \mathbf{H}_{e,t}^i\right)\Vert$\end{small} is similarly
\begin{equation}\small
\begin{aligned}
&\Vert\nabla R\left(\mathbf{\Sigma}_{s}, \mathbf{H}_{e,t}^i\right)\Vert\\
&=\Vert \rho_e\mathbf{G}\mathbf{H}_{e,t}^i\left(\mathbf{I}+\rho_e\mathbf{H}_{e,t}^{iH} \mathbf{G}^H(\mathbf{\Sigma}_s^n+ \mathbf{\Sigma}_z) \mathbf{G}\mathbf{H}_{e,t}^i \right)^{-1}\\
&\quad\mathbf{H}_{e,t}^{iH} \mathbf{G}^Hp(\mathbf{H}_{e,t}^i)\Vert\\
&\leq\rho_e\Vert\mathbf{G}^H\mathbf{G}\Vert\Vert\mathbf{H}_{e,t}^i\mathbf{H}_{e,t}^{iH}\Vert\\
&\quad\Vert\left(\mathbf{I}+\rho_e\mathbf{H}_{e,t}^{iH} \mathbf{G}^H(\mathbf{\Sigma}_s+ \mathbf{\Sigma}_z) \mathbf{G}\mathbf{H}_{e,t}^i \right)^{-1}\Vert p(\mathbf{H}_{e,t}^i)\\
&=\rho_e G N_e \Vert\mathbf{H}_{e,t}^i\mathbf{H}_{e,t}^{iH}\Vert p(\mathbf{H}_{e,t}^i)
\end{aligned}
\end{equation}
Following the same analysis as steps (\ref{step1})-(\ref{step5}), we know that the functions \begin{small}$R\left(\mathbf{\Sigma}_{z}, \mathbf{H}_{e,t}^i\right)$\end{small} and \begin{small}$\rho_e G N_e \Vert\mathbf{H}_{e,t}^i\mathbf{H}_{e,t}^{iH}\Vert p(\mathbf{H}_{e,t}^i)$\end{small} are both integrable and \begin{small}$\nabla R\left(\mathbf{\Sigma}_{z}, \mathbf{H}_{e,t}^i\right)$\end{small} exists. Then, the interchange of the gradient and the integral is allowed in \begin{small}$R\left(\mathbf{\Sigma}_{z}, \mathbf{H}_{e,t}^i\right)$\end{small} \cite[Theorem A.1]{Measure Theory}, i.e., \begin{small}$\Expectation[\nabla f_{e,t}^i(\mathbf{\Sigma}_{s}^n+\mathbf{\Sigma}_{z})]=\nabla\Expectation[ f_{e,t}^i(\mathbf{\Sigma}_{s}^n+\mathbf{\Sigma}_{z})]$\end{small}. Similarly, \begin{small}$\Expectation\left[\nabla f_{e,t}^i\left(\mathbf{\Sigma}_{z}\right)\right]=\nabla\Expectation\left[ f_{e,t}^i\left(\mathbf{\Sigma}_{z}\right)\right]$\end{small}. Therefore,
\begin{equation}\small
\begin{aligned}
&\Expectation\left[\nabla_{\mathbf{\Sigma}_s^{*}}C_3(\mathbf{\Sigma}_s^{t}, \mathbf{\Sigma}_z^{t}, \mathbf{\Theta}^t,\mathbf{H}_{e,t}^i)-\nabla_{\mathbf{\Sigma}_s^{*}}C_{3}\left(\mathbf{\Sigma}_s^{t},\mathbf{\Sigma}_z^{t},\mathbf{\Theta}^t\right)\right]=0
\end{aligned}
\end{equation}
Similarly,
\begin{equation}\small
\Expectation\left[\nabla_{\mathbf{\Sigma}_z^{*}}C_3(\mathbf{\Sigma}_s^{t}, \mathbf{\Sigma}_z^{t}, \mathbf{\Theta}^t,\mathbf{H}_{e,t}^i)-\nabla_{\mathbf{\Sigma}_z^{*}}C_{3}\left(\mathbf{\Sigma}_s^{t},\mathbf{\Sigma}_z^{t},\mathbf{\Theta}^t\right)\right]=0
\end{equation}.

Following the same analysis as (\ref{Estimation1}) for any \begin{small}$\mathbf{\Sigma}_{s}^n$\end{small} and \begin{small}$\mathbf{\Theta}^n$\end{small}, we obtain
\begin{equation}\small
\label{up1}
\begin{aligned}
&\Vert\nabla f_{r,t}\left(\mathbf{\Sigma}_{s}^n+\mathbf{\Sigma}_{z}^1\right)-\nabla f_{r,t}\left(\mathbf{\Sigma}_{s}^n+\mathbf{\Sigma}_{z}^2\right)\Vert\\
&\leq\left(\mathop{\max}_{\mathbf{\Sigma}_{z}} \Vert \nabla f_{r,t}\left(\mathbf{\Sigma}_{s}^n+\mathbf{\Sigma}_{z}\right)\Vert\right)^2\Vert\mathbf{\Sigma}_{z}^1-\mathbf{\Sigma}_{z}^2\Vert\\
&\leq\left(\rho_r G {{N_I}} \mathbf{h}_r^H\mathbf{h}_r\right)^2\Vert\mathbf{\Sigma}_{z}^1-\mathbf{\Sigma}_{z}^2\Vert
\end{aligned}
\end{equation}
where the last equality is due to 
\begin{equation}\small
\begin{aligned}
&\bigg\Vert\frac{\rho_r \mathbf{G} \mathbf{\Theta}^{nH} \mathbf{h}_r\mathbf{h}_r^H \mathbf{\Theta}^n \mathbf{G}^H}{1+\rho_r\mathbf{h}_r^H \mathbf{\Theta}^n \mathbf{G}^H\left(\mathbf{\Sigma}_{s}^n+\mathbf{\Sigma}_{z}\right) \mathbf{G} \mathbf{\Theta}^{nH} \mathbf{h}_r}\bigg\Vert\\
&\leq \Vert\rho_r \mathbf{G} \mathbf{\Theta}^{nH} \mathbf{h}_r\mathbf{h}_r^H \mathbf{\Theta}^n \mathbf{G}^H\Vert\\
&\leq \rho_r G N_I \mathbf{h}_r^H\mathbf{h}_r.
\end{aligned}
\end{equation}
and
\begin{eqnarray}\small
\label{up}
&&\Vert\nabla\Expectation\left[ f_{e,t}^i\left(\mathbf{\Sigma}_{s}^n+\mathbf{\Sigma}_{z}^1\right)\right]-\nabla\Expectation\left[ f_{e,t}^i\left(\mathbf{\Sigma}_{s}^n+\mathbf{\Sigma}_{z}^2\right)\right]\Vert \nonumber\\
&&=\Vert\Expectation\left[\nabla f_{e,t}^i\left(\mathbf{\Sigma}_{s}^n+\mathbf{\Sigma}_{z}^1\right)\right]-\Expectation\left[\nabla f_{e,t}^i\left(\mathbf{\Sigma}_{s}^n+\mathbf{\Sigma}_{z}^2\right)\right]\Vert \nonumber\\
&&\leq \left(\mathop{\max}_{\mathbf{\Sigma}_{z}} \Expectation\left[\Vert\nabla f_e^i\left(\mathbf{\Sigma}_{s}^n+\mathbf{\Sigma}_{z}\right)\Vert\right]\right)^2\Vert\mathbf{\Sigma}_{z}^1-\mathbf{\Sigma}_{z}^2\Vert\\
&&\leq 4\rho_e^2 G^2 N_I^4 N_e^6\Vert\mathbf{\Sigma}_{z}^1-\mathbf{\Sigma}_{z}^2\Vert.\nonumber
\end{eqnarray}
Similarly, \begin{small}$\Vert\nabla f_{r,t}\left(\mathbf{\Sigma}_{z}^1\right)-\nabla f_{r,t}\left(\mathbf{\Sigma}_{z}^2\right)\Vert\leq\left(\rho_r G N_I \mathbf{h}_r^H\mathbf{h}_r\right)^2\Vert\mathbf{\Sigma}_{z}^1-\mathbf{\Sigma}_{z}^2\Vert$\end{small} and \begin{small}$\Vert\Expectation\left[\nabla f_{e,t}^i\left(\mathbf{\Sigma}_{z}^1\right)\right]-\Expectation\left[\nabla f_{e,t}^i\left(\mathbf{\Sigma}_{z}^2\right)\right]\Vert\leq4\rho_e^2 G^2 N_I^4 N_e^6\Vert\mathbf{\Sigma}_{z}^1-\mathbf{\Sigma}_{z}^2\Vert$\end{small}. Then, the gradient of \begin{small}$C_{3}\left(\mathbf{\Sigma}_s,\mathbf{\Sigma}_z,\mathbf{\Theta}\right)$\end{small} is Lipschitz continuous with respect to \begin{small}$\mathbf{\Sigma}_z$\end{small}. The above property is also true for \begin{small}$\mathbf{\Sigma}_s$\end{small}. As the upper bounds of (\ref{up1}) and (\ref{up}) are independent of $\mathbf{\Sigma}_{s}^n$ and \begin{small}$\mathbf{\Theta}^n$\end{small}, we can find a Lipschitz constant independent of \begin{small}$\mathbf{\Sigma}_{s}$\end{small}, \begin{small}$\mathbf{\Sigma}_{z}$\end{small} and \begin{small}$\mathbf{\Theta}$\end{small}. So, for any $t$, \begin{small}$\Vert\nabla_{\mathbf{\Sigma}_z^{*}}C_{3}\left(\mathbf{\Sigma}_s^{t},\mathbf{\Sigma}_z^{t},\mathbf{\Theta}^t\right)\Vert$\end{small} and \begin{small}$\Vert\nabla_{\mathbf{\Sigma}_z^{*}}C_{3}\left(\mathbf{\Sigma}_s^{t},\mathbf{\Sigma}_z^{t},\mathbf{\Theta}^t\right)\Vert$\end{small} are bounded. Therefore, there exists a constant \begin{small}$\sigma$\end{small} such that
\begin{equation}\small
\begin{aligned}
&\Expectation\left[\Vert\nabla_{\mathbf{\Sigma}_s^{*}}C_3(\mathbf{\Sigma}_s^{t}, \mathbf{\Sigma}_z^{t}, \mathbf{\Theta}^t,\mathbf{H}_{e,t}^i)-\nabla_{\mathbf{\Sigma}_s^{*}}C_{3}\left(\mathbf{\Sigma}_s^{t},\mathbf{\Sigma}_z^{t},\mathbf{\Theta}^t\right)\Vert^2\right]\\
&\leq 2\Expectation\bigg[\Vert\nabla_{\mathbf{\Sigma}_s^{*}}C_3(\mathbf{\Sigma}_s^{t}, \mathbf{\Sigma}_z^{t}, \mathbf{\Theta}^t,\mathbf{H}_{e,t}^i)\Vert^2\\
&\qquad\qquad\qquad\qquad+\Vert\nabla_{\mathbf{\Sigma}_s^{*}}C_{3}\left(\mathbf{\Sigma}_s^{t},\mathbf{\Sigma}_z^{t},\mathbf{\Theta}^t\right)\Vert^2\bigg]<\sigma^2\\
&\Expectation\left[\Vert\nabla_{\mathbf{\Sigma}_z^{*}}C_3(\mathbf{\Sigma}_s^{t}, \mathbf{\Sigma}_z^{t}, \mathbf{\Theta}^t,\mathbf{H}_{e,t}^i)-\nabla_{\mathbf{\Sigma}_z^{*}}C_{3}\left(\mathbf{\Sigma}_s^{t},\mathbf{\Sigma}_z^{t},\mathbf{\Theta}^t\right)\Vert^2\right]\\
&\leq2\Expectation\bigg[\Vert\nabla_{\mathbf{\Sigma}_z^{*}}C_3(\mathbf{\Sigma}_s^{t}, \mathbf{\Sigma}_z^{t}, \mathbf{\Theta}^t,\mathbf{H}_{e,t}^i)\Vert^2\\
&\qquad\qquad\qquad\qquad+\Vert\nabla_{\mathbf{\Sigma}_z^{*}}C_{3}\left(\mathbf{\Sigma}_s^{t},\mathbf{\Sigma}_z^{t},\mathbf{\Theta}^t\right)\Vert^2\bigg]<\sigma^2.
\end{aligned}
\end{equation}

\section{The Projected Gradient of $\left(\mathbf\Sigma_{s}^\infty, \mathbf\Sigma_{z}^\infty\right)$ in Algorithm 3}
\label{Convergence of the Hybrid SPG-CP Algorithm}
In order to discuss the projected gradient of \begin{small}$\left(\mathbf\Sigma_{s}^\infty, \mathbf\Sigma_{z}^\infty\right)$\end{small} in Algorithm 3, we regard each \begin{small}$\mathbf{H}_{e,t}^i$\end{small} as a random matrix with the same distribution as \begin{small}$\mathbf{H}_e$\end{small} in (\ref{rate3}). Some important stochastic projected gradients and gaps over \begin{small}$\mathcal{H}_e^t$\end{small} associated with \begin{small}$\overline{C}_3(\mathbf{\Sigma}_s^{t}, \mathbf{\Sigma}_z^{t}, \mathbf{\Theta}^t, (\mathbf{H}_{e,t}^i)_{i=1}^{\lceil t^{\alpha}\rceil})$\end{small} are given by
\begin{subequations}\small
\label{gradient_definition}
	\begin{align}
	(\widetilde{\mathbf{W}}_s^t, \widetilde{\mathbf{W}}_z^t)&\triangleq P_{\mathcal{X}_2}[(\mathbf\Sigma_{s}^t,\mathbf\Sigma_{z}^t),(-\nabla_{\mathbf{\Sigma}_{s}^*}\overline{C}_3(\mathbf{\Sigma}_s^{t}, \mathbf{\Sigma}_z^{t}, \mathbf{\Theta}^t, (\mathbf{H}_{e,t}^i)_{i=1}^{\lceil t^{\alpha}\rceil}),\nonumber\\
	&\qquad\qquad-\nabla_{\mathbf{\Sigma}_{z}^*}\overline{C}_3(\mathbf{\Sigma}_s^{t}, \mathbf{\Sigma}_z^{t}, \mathbf{\Theta}^t, (\mathbf{H}_{e,t}^i)_{i=1}^{\lceil t^{\alpha}\rceil})), r]\\
	(\mathbf{W}_s^t, \mathbf{W}_z^t)&\triangleq P_{\mathcal{X}_2}[(\mathbf\Sigma_{s}^t,\mathbf\Sigma_{z}^t),(-\nabla_{\mathbf{\Sigma}_{s}^*}C_3(\mathbf{\Sigma}_s^{t}, \mathbf{\Sigma}_z^{t}, \mathbf{\Theta}^t),\nonumber\\
	&\qquad\qquad-\nabla_{\mathbf{\Sigma}_{z}^*}C_3(\mathbf{\Sigma}_s^{t}, \mathbf{\Sigma}_z^{t}, \mathbf{\Theta}^t)), r]\\
	\mathbf{D}_s^t &\triangleq-\nabla_{\mathbf{\Sigma}_{s}^*}\overline{C}_3(\mathbf{\Sigma}_s^{t}, \mathbf{\Sigma}_z^{t}, \mathbf{\Theta}^t, (\mathbf{H}_{e,t}^i)_{i=1}^{\lceil t^{\alpha}\rceil})\nonumber\\
	&\qquad\qquad+\nabla_{\mathbf{\Sigma}_{s}^*}C_3(\mathbf{\Sigma}_s^{t}, \mathbf{\Sigma}_z^{t}, \mathbf{\Theta}^t)\\
	\mathbf{D}_z^t &\triangleq-\nabla_{\mathbf{\Sigma}_{z}^*}\overline{C}_3(\mathbf{\Sigma}_s^{t}, \mathbf{\Sigma}_z^{t}, \mathbf{\Theta}^t, (\mathbf{H}_{e,t}^i)_{i=1}^{\lceil t^{\alpha}\rceil})\nonumber\\
	&\qquad\qquad+\nabla_{\mathbf{\Sigma}_{z}^*}C_3(\mathbf{\Sigma}_s^{t}, \mathbf{\Sigma}_z^{t}, \mathbf{\Theta}^t).
	\end{align}
\end{subequations}

For any given \begin{small}$t>0$\end{small}, Proposition \ref{stateL1} yields:
\begin{equation}\small
\label{mainequation}
\begin{aligned}
&-C_3(\mathbf{\Sigma}_s^{t+1}, \mathbf{\Sigma}_z^{t+1}, \mathbf{\Theta}^{t+1})\leq -C_3(\mathbf{\Sigma}_s^{t+1}, \mathbf{\Sigma}_z^{t+1}, \mathbf{\Theta}^t)\\
&\leq -C_3(\mathbf{\Sigma}_s^{t}, \mathbf{\Sigma}_z^{t}, \mathbf{\Theta}^t)-2\langle \nabla_{\mathbf{\Sigma}_s^{*}} C_3(\mathbf{\Sigma}_s^{t}, \mathbf{\Sigma}_z^{t}, \mathbf{\Theta}^t),\mathbf{\Sigma}_s^{t+1}-\mathbf{\Sigma}_s^{t}\rangle\\
&\quad+L\Vert\mathbf{\Sigma}_s^{t+1}-\mathbf{\Sigma}_s^{t}\Vert^2-2\langle \nabla_{\mathbf{\Sigma}_z^{*}} C_3(\mathbf{\Sigma}_s^{t}, \mathbf{\Sigma}_z^{t}, \mathbf{\Theta}^t), \mathbf{\Sigma}_z^{t+1}-\mathbf{\Sigma}_z^{t}\rangle\\
&\quad+L\Vert\mathbf{\Sigma}_z^{t+1}-\mathbf{\Sigma}_z^{t}\Vert^2\\
&= -C_3(\mathbf{\Sigma}_s^{t}, \mathbf{\Sigma}_z^{t}, \mathbf{\Theta}^t)-2r\langle -\nabla_{\mathbf{\Sigma}_s^{*}} \overline{C}_3(\mathbf{\Sigma}_s^{t}, \mathbf{\Sigma}_z^{t}, \mathbf{\Theta}^t,(\mathbf{H}_{e,t}^i)_{i=1}^{\lceil t^{\alpha}\rceil}), \\
&\quad\widetilde{\mathbf{W}}_s^t\rangle+ 2r\langle \mathbf{D}_s^t, \widetilde{\mathbf{W}}_s^t\rangle+Lr^2\Vert\widetilde{\mathbf{W}}_s^t\Vert^2-2r\langle -\nabla_{\mathbf{\Sigma}_z^{*}} \overline{C}_3(\mathbf{\Sigma}_s^{t}, \mathbf{\Sigma}_z^{t}, \\
&\quad\mathbf{\Theta}^t, (\mathbf{H}_{e,t}^i)_{i=1}^{\lceil t^{\alpha}\rceil}), \widetilde{\mathbf{W}}_z^t\rangle+ 2r\langle \mathbf{D}_z^t, \widetilde{\mathbf{W}}_z^t\rangle+Lr^2\Vert\widetilde{\mathbf{W}}_z^t\Vert^2\\
&\overset{(e)}{\leq} -C_3(\mathbf{\Sigma}_s^{t}, \mathbf{\Sigma}_z^{t},\mathbf{\Theta}^t)-(2r-Lr^2)\left(\Vert\widetilde{\mathbf{W}}_s^t\Vert^2+\Vert\widetilde{\mathbf{W}}_z^t\Vert^2\right)\\
&\quad+ 2r\left(\langle \mathbf{D}_s^t, \widetilde{\mathbf{W}}_s^t-\mathbf{W}_s^t\rangle+\langle \mathbf{D}_z^t, \widetilde{\mathbf{W}}_z^t-\mathbf{W}_z^t\rangle\right)\\
&\quad+ 2r\left(\langle \mathbf{D}_s^t, \mathbf{W}_s^t\rangle+\langle \mathbf{D}_z^t, \mathbf{W}_z^t\rangle\right)\\
&\overset{(f)}{\leq} -C_3(\mathbf{\Sigma}_s^{t}, \mathbf{\Sigma}_z^{t}, \mathbf{\Theta}^t)-(2r-Lr^2)\left(\Vert\widetilde{\mathbf{W}}_s^t\Vert^2+\Vert\widetilde{\mathbf{W}}_z^t\Vert^2\right)\\
&\quad+ 2r\left(\langle \mathbf{D}_s^t, \mathbf{W}_s^t\rangle+\langle \mathbf{D}_z^t, \mathbf{W}_z^t\rangle\right)+ 2r\left(\Vert \mathbf{D}_s^t\Vert^2+\Vert \mathbf{D}_z^t\Vert^2\right)
\end{aligned}
\end{equation}
where (e) and (f) follow from Lemmas \ref{Projection1} and \ref{Projection3}, respectively.
Notice that \begin{small}$\mathbf{\Sigma}_s^t$\end{small}, \begin{small}$\mathbf{\Sigma}_z^t$\end{small}, and \begin{small}$\mathbf{\Theta}^t$\end{small} are the functions of the history \begin{small}$\mathcal{H}_e^{[t-1]}=\left\{\mathcal{H}_e^{i}: 1\leq i\leq t-1\right\}$\end{small} of the generated random process and hence are random. \begin{small}$\Expectation_{\mathcal{H}_e^{t}}\left(\langle \mathbf{D}_s^t, \mathbf{W}_s^t\rangle+\langle \mathbf{D}_z^t, \mathbf{W}_z^t\rangle\vert\mathcal{H}_e^{[t-1]}\right)=0$\end{small} holds due to (\ref{assumptions}). Thus, \begin{small}$\Expectation_{\mathcal{H}_e^{[t]}}\left(\langle \mathbf{D}_s^t, \mathbf{W}_s^t\rangle+\langle \mathbf{D}_z^t, \mathbf{W}_z^t\rangle\right)=0, \forall t=1,2,...,N$\end{small}. In addition, due to
\begin{equation}\small
\label{DD}
\begin{aligned}
\Expectation_{\mathcal{H}_e^{t}}&\left(\Vert \mathbf{D}_s^t\Vert^2+\Vert \mathbf{D}_z^t\Vert^2\vert\mathcal{H}_e^{[t-1]}\right)\\
\leq&\frac{1}{{\lceil t^{\alpha}\rceil}^2}\sum_{i=1}^{\lceil t^{\alpha}\rceil}\Expectation_{\mathcal{H}_e^{t}}(\Vert \nabla_{\mathbf{\Sigma}_{s}^*}C_3(\mathbf{\Sigma}_s^{t}, \mathbf{\Sigma}_z^{t}, \mathbf{\Theta}^t,\mathbf{H}_{e,t}^i)\\
&-\nabla_{\mathbf{\Sigma}_{s}^*}C_3(\mathbf{\Sigma}_s^{t}, \mathbf{\Sigma}_z^{t}, \mathbf{\Theta}^t)\Vert^2\vert\mathcal{H}_e^{[t-1]})\\
&+\frac{1}{{\lceil t^{\alpha}\rceil}^2}\sum_{i=1}^{\lceil t^{\alpha}\rceil}\Expectation_{\mathcal{H}_e^{t}}(\Vert \nabla_{\mathbf{\Sigma}_{z}^*}C_3(\mathbf{\Sigma}_s^{t}, \mathbf{\Sigma}_z^{t}, \mathbf{\Theta}^t,\mathbf{H}_{e,t}^i)\\
&-\nabla_{\mathbf{\Sigma}_{z}^*}C_3(\mathbf{\Sigma}_s^{t}, \mathbf{\Sigma}_z^{t}, \mathbf{\Theta}^t)\Vert^2\vert\mathcal{H}_e^{[t-1]})\\
\leq&\frac{2\sigma^2}{\lceil t^{\alpha}\rceil}\leq\frac{2\sigma^2}{ t^{\alpha}}
\end{aligned}
\end{equation}
\begin{small}$\Expectation_{\mathcal{H}_e^{[t]}}\left(\Vert \mathbf{D}_s^t\Vert^2+\Vert \mathbf{D}_z^t\Vert^2\right)\leq\frac{2\sigma^2}{t^{\alpha}}$\end{small} holds.
Taking expectations with respect to \begin{small}$\mathcal{H}_e^{[t]}$\end{small} on both sides of the inequality (\ref{mainequation}), we have
\begin{equation}\small
\begin{aligned}
&\Expectation_{\mathcal{H}_e^{[t]}}\left\{ \Vert\widetilde{\mathbf{W}}_s^t\Vert^2+\Vert\widetilde{\mathbf{W}}_z^t\Vert^2\right\}\\
& \leq \frac{C_3(\mathbf{\Sigma}_s^{t+1}, \mathbf{\Sigma}_z^{t+1}, \mathbf{\Theta}^{t+1}) - C_3(\mathbf{\Sigma}_s^t, \mathbf{\Sigma}_z^t, \mathbf{\Theta}^t)}{2r-Lr^2}+\frac{4\sigma^2}{{t^{\alpha}}(2-Lr)}
\end{aligned}
\end{equation}
 Furthermore, considering the inequality
\begin{equation}\small
\begin{aligned}
&\Expectation_{\mathcal{H}_e^{[t]}}\left\{\Vert \mathbf{W}_s^t\Vert^2+\Vert \mathbf{W}_z^t\Vert^2\right\}\\
&=\Expectation_{\mathcal{H}_e^{[t]}}\left\{\Vert\mathbf{W}_s^t-\widetilde{\mathbf{W}}_s^t+\widetilde{\mathbf{W}}_s^t\Vert^2+\Vert\mathbf{W}_z^t-\widetilde{\mathbf{W}}_z^t+\widetilde{\mathbf{W}}_z^t\Vert^2\right\}\\
&\leq 2\Expectation_{\mathcal{H}_e^{[t]}}\left\{\Vert\mathbf{W}_s^t-\widetilde{\mathbf{W}}_s^t\Vert^2+\Vert\mathbf{W}_z^t-\widetilde{\mathbf{W}}_z^t\Vert^2\right\}\\
&\quad+2\Expectation_{\mathcal{H}_e^{[t]}}\left\{\Vert\widetilde{\mathbf{W}}_s^t\Vert^2+\Vert\widetilde{\mathbf{W}}_z^t\Vert^2\right\}\\
&\leq 2\Expectation_{\mathcal{H}_e^{[t]}}\left\{\Vert\mathbf{D}_s^t\Vert^2+\Vert\mathbf{D}_z^t\Vert^2\right\}+2\Expectation_{\mathcal{H}_e^{[t]}}\left\{\Vert\widetilde{\mathbf{W}}_s^t\Vert^2+\Vert\widetilde{\mathbf{W}}_z^t\Vert^2\right\}\\
&\leq \frac{4\sigma^2}{t^{\alpha}}+2\Expectation_{\mathcal{H}_e^{[t]}}\left\{\Vert\widetilde{\mathbf{W}}_s^t\Vert^2+\Vert\widetilde{\mathbf{W}}_z^t\Vert^2\right\}.
\end{aligned}
\end{equation}
Adding the above over \begin{small}$t=1,2,\cdots, N$\end{small}, we have
\begin{equation}\small
\begin{aligned}
&\sum_{t=1}^N\Expectation_{\mathcal{H}_e^{[t]}}\left\{\Vert \mathbf{W}_s^t\Vert^2+\Vert \mathbf{W}_z^t\Vert^2\right\}\\
&\leq \frac{C_3(\mathbf{\Sigma}_s^{N+1}, \mathbf{\Sigma}_z^{N+1}, \mathbf{\Theta}^{N+1}) - C_3(\mathbf{\Sigma}_s^1, \mathbf{\Sigma}_z^1, \mathbf{\Theta}^1)}{r-Lr^2/2}\\
&\quad+\sum_{t=1}^N\left[\frac{8\sigma^2}{t^\alpha(2-Lr)}+\frac{4\sigma^2}{t^\alpha} \right]
\end{aligned}
\end{equation}
As \begin{small}$N\to\infty$\end{small}, since \begin{small}$\alpha>1$\end{small}, we have
\begin{equation}\small
\begin{aligned}
&\sum_{t=1}^\infty\Expectation_{\mathcal{H}_e^{[t]}}\left\{\Vert \mathbf{W}_s^t\Vert^2+\Vert \mathbf{W}_z^t\Vert^2\right\}\\
&\leq \frac{C_3(\widehat{\mathbf{\Sigma}}_s, \widehat{\mathbf{\Sigma}}_z, \widehat{\mathbf{\Theta}}) - C_3(\mathbf{\Sigma}_s^1, \mathbf{\Sigma}_z^1, \mathbf{\Theta}^1)}{r-Lr^2/2}\\
&\quad+\left(\frac{8\sigma^2}{2-Lr}+4\sigma^2\right)\left(1+\frac{1}{\alpha-1}\right)
\end{aligned}
\end{equation}
where \begin{small}$(\widehat{\mathbf{\Sigma}}_s, \widehat{\mathbf{\Sigma}}_z, \widehat{\mathbf{\Theta}})$\end{small} is an optimal solution to \begin{small}$C_3(\mathbf{\Sigma}_s, \mathbf{\Sigma}_z, \mathbf{\Theta})$\end{small}. That implies \begin{small}$\lim_{t \to \infty}\Expectation_{\mathcal{H}_e^{[t]}}\{\Vert \mathbf{W}_s^t\Vert^2+\Vert \mathbf{W}_z^t\Vert^2\}=0$\end{small}.

\section{Proof of Lemma \ref{Int_log}}
\label{Appendix_Int_log}
The eigenvalues\footnote{Here, we think the singular matrix is also invertible and the inverse matrix has some \begin{small}$\infty$\end{small} as eigenvalues.} of \begin{small}$\left(\rho\sigma^2\mathbf{Q}\right)^{-1}$\end{small} are \begin{small}$\{\infty,\infty...,\infty, \tilde{t}' ,\tilde{t}',...,\tilde{t}'\}$\end{small}, where the number of \begin{small}$\tilde{t}'=\frac{1}{\tilde{t}}$\end{small} is \begin{small}$N_1$\end{small}. From \cite[Eqn. (53)]{calculation_expctation}, the p.d.f of \begin{small}$I = \log\left(1+\mathbf{z}^H\rho\mathbf{Q}\mathbf{z}\right)$\end{small} and its expectation \begin{small}$\left \langle I \right\rangle = \Expectation_{\mathbf{z}} \left\{\log\left(1+\mathbf{z}^H\rho\mathbf{Q}\mathbf{z}\right)\right\}$\end{small} are shown as follows:
\begin{equation}\small
\begin{aligned}
\pdf(I)&=e^I \prod_{i=1}^{N_1} \left(-j a_i\right)\int \frac{dk}{2\pi}\frac{e^{jk\left(e^I-1\right)}}{\prod\limits_{i=1}^{N_1}\left(k-ja_i\right)} \\
&= e^I\left(-j \tilde{t}'\right)^{N_1}\int \frac{dk}{2\pi}\frac{e^{jk\left(e^I-1\right)}}{\left(k-j\tilde{t}'\right)^{N_1}}\\
&=\frac{e^I}{2\pi} \left(-j \tilde{t}'\right)^{N_1} * 2\pi j \boldsymbol{\Res}\left(\frac{e^{jk\left(e^I-1\right)}}{\left(k-j\tilde{t}'\right)^{N_1}},\,j\tilde{t}'\right)\\
&
=\frac{\tilde{t}'^{N_1} \left(e^I - 1\right)^{N_1-1} e^{I-\tilde{t}'\left(e^I - 1\right)}}{\left(N_1-1\right)!}\\
\end{aligned}
\end{equation}
\begin{equation}\small
\begin{aligned}
\left \langle I \right\rangle& = \frac{\int_{0}^{\infty}I \tilde{t}'^{N_1} \left(e^I - 1\right)^{N_1-1} e^{I-\tilde{t}'\left(e^I - 1\right)}dI}{\left(N_1-1\right)!}\\
&=\frac{\int_{0}^{\infty}\log\left(1+\frac{x}{\tilde{t}'}\right)x^{N_1-1}e^{-x}dx}{\left(N_1-1\right)!}\\
&=\frac{\int_{0}^{\infty}\log\left(1+\tilde{t}x\right)x^{N_1-1}e^{-x}dx}{\left(N_1-1\right)!}\triangleq F_1\left(\tilde{t},\,N_1\right).
\end{aligned}
\end{equation}
The first- and second-order derivatives \cite[\S 12.212]{Table_of_Integrals_Series_and_Products} of \begin{small}$F_1\left(\tilde{t},\,N_1\right)$\end{small} with respective to \begin{small}$\tilde{t}$\end{small} are given by:
\begin{equation}\small
\begin{aligned}
\frac{\partial F_1(\tilde{t},N_1)}{\partial \tilde{t}} &= \frac{\int_{0}^{\infty}{\frac{1}{1+\tilde{t}x}x^{N_1}e^{-x}}dx}{\left(N_1-1\right)!}\\
\frac{\partial^2 F_1(\tilde{t},N_1)}{\partial \tilde{t}^2} &= \frac{\int_{0}^{\infty}{\frac{-1}{\left(1+\tilde{t}x\right)^2}x^{N_1+1}e^{-x}}dx}{\left(N_1-1\right)!}.
\end{aligned}
\end{equation}
\section{Proof of Proposition \ref{rankone}}
\label{Appendix_rankone}
For any given \begin{small}$\mathbf{\Theta}$\end{small}, the necessary conditions for the optimal \begin{small}$\boldsymbol{\Sigma}_s$\end{small} can be obtained based on the KKT conditions, Let us construct the cost function as follows:
\begin{equation}\small
\begin{aligned} 
\mathcal{L}\left(\mathbf{\Sigma}_s, \lambda, \mathbf{\Psi}_s\right) = -C_{1}(\mathbf{\Sigma}_s, \mathbf{\Theta}) + \lambda\left[\Tr\left(\mathbf{\Sigma}_s\right)-1\right] - \Tr\left(\mathbf{\Psi}_s\mathbf{\Sigma}_s\right)
\end{aligned}
\end{equation}
where \begin{small}$\lambda\geq 0, \mathbf{\Psi}_s \succeq 0$\end{small} are the Lagrange multipliers accounting for the total power constraint and the constraint that \begin{small}$\mathbf{\Sigma}_s$\end{small} is positive semidefinite. Then the KKT conditions enable us to write
\begin{subequations}\small
	\begin{align}
	&-\mathbf{A}+\lambda\mathbf{I} - \mathbf{\Psi}_s = \mathbf{0}, \mathbf{\Psi}_s\mathbf{\Sigma}_s = \mathbf{\Sigma}_s\mathbf{\Psi}_s = 0\\
	&\Tr\left(\boldsymbol{\Sigma}_s\right)\leq 1,\,\,\,\,\lambda\left[\Tr\left(\mathbf{\Sigma}_s\right)-1\right] = 0\\
	&\boldsymbol{\Psi}_s \succeq 0,\,\,\,\, \boldsymbol{\Sigma}_s \succeq 0
	\end{align}
\end{subequations}
with
\begin{equation}\small
\begin{aligned}
&\mathbf{A}=\frac{\rho_r\mathbf{G} \mathbf{\Theta}^H \mathbf{h}_r \mathbf{h}_r^H \mathbf{\Theta} \mathbf{G}^H}{1+\rho_r\mathbf{h}_r^H \mathbf{\Theta} \mathbf{G}^H\mathbf{\Sigma}_s \mathbf{G} \mathbf{\Theta}^H \mathbf{h}_r}\\
&\quad-\Expectation_{\mathbf{H}_e}\left[\rho_e\mathbf{G} \mathbf{H}_e\left(\mathbf{I}+\rho_e\mathbf{H}_e^H \mathbf{G}^H\mathbf{\Sigma}_s \mathbf{G} \mathbf{H}_e\right)^{-1}\mathbf{H}_e^H \mathbf{G}^H\right].
\end{aligned}
\end{equation}
The optimal \begin{small}$\mathbf{\Sigma}_s$\end{small} satisfies \begin{small}$\mathbf{A}\mathbf{\Sigma}_s=\mathbf{\Sigma}_s\mathbf{A} = \lambda \mathbf{\Sigma}_s$\end{small}. Matrix \begin{small}$\mathbf{A}$\end{small} has at most one positive eigenvalue, because the first term of \begin{small}$\mathbf{A}$\end{small} is rank one with one positive eigenvalue and the second term is non-negative definite matrix. Therefore, \begin{small}$\mathbf{\Sigma}_s$\end{small} is rank one.

\end{document}